\documentclass[showpacs,twocolumn,prx,superscriptaddress,floatfix]{revtex4-1}
\usepackage{hyperref}
\usepackage{qcircuit}
\usepackage[dvips]{graphicx}
\usepackage{amsmath,amssymb,amsthm,mathrsfs,amsfonts,dsfont}
\usepackage{multirow}
\usepackage{subfigure, epsfig}
\usepackage{braket}
\usepackage{bm}
\usepackage{enumerate}
\usepackage{color}
\usepackage[dvipsnames]{xcolor}

\usepackage{newtxtext,newtxmath}

\newtheorem{theorem}{Theorem}
\newtheorem*{conjecture*}{Conjecture}
\newtheorem{lemma}{Lemma}

\newtheorem{corollary}{Corollary}

\newtheorem{remark}{Remark}
\newcommand{\LR}{\mathrm{LR}}
\newcommand{\mc}[1]{\mathcal{#1}}
\newcommand{\tr}{\mathrm{Tr}}
\newcommand{\ketbra}[2]{\vert #1 \rangle \langle #2 \vert}
\newcommand{\yx}[1]{\textcolor{black}{#1}}

\newcommand{\jt}[1]{\textcolor{black}{#1}}
\newcommand{\mg}[1]{\textcolor{black}{#1}}

\begin{document}

\title{Universal and Operational Benchmarking of Quantum Memories}

\begin{abstract}
Quantum memory --- the capacity to store and faithfully recover unknown quantum states --- is essential for quantum-enhanced technology. There is thus a pressing need for operationally meaningful means to benchmark candidate memories across diverse physical platforms. Here we introduce a universal benchmark distinguished by its relevance across multiple key operational settings, exactly quantifying (1) the memory's robustness to noise, (2) the number of noiseless qubits needed for its synthesis, (3) its potential to speed up statistical sampling tasks, and (4) performance advantage in non-local games beyond classical limits. The measure is analytically computable for low-dimensional systems and can be efficiently bounded in experiment without tomography. We thus illustrate quantum memory as a meaningful resource, with our benchmark reflecting both its cost of creation and what it can accomplish. We demonstrate the benchmark on the five-qubit IBM Q hardware, and apply it to witness efficacy of error-suppression techniques and quantify non-Markovian noise. We thus present an experimentally accessible, practically meaningful, and universally relevant quantifier of a memory's capability to preserve quantum advantage.
\end{abstract}
\date{\today}

\author{Xiao Yuan}
\email{xiao.yuan.ph@gmail.com}
\affiliation{Department of Materials, University of Oxford, Parks Road, Oxford OX1 3PH, United Kingdom}

\author{Yunchao Liu}
\email{yunchaoliu@berkeley.edu}
\affiliation{Department of Electrical Engineering and Computer Sciences, University of California, Berkeley, California 94720, USA}
\affiliation{Center for Quantum Information, Institute for Interdisciplinary Information Sciences, Tsinghua University, Beijing 100084, China}

\author{Qi Zhao}
\affiliation{Center for Quantum Information, Institute for Interdisciplinary Information Sciences, Tsinghua University, Beijing 100084, China}

\author{Bartosz Regula}
\affiliation{School of Physical and Mathematical Sciences, Nanyang Technological University, 637371, Singapore}
\affiliation{Complexity Institute, Nanyang Technological University, 637335, Singapore}

\author{Jayne Thompson}
\affiliation{Centre for Quantum Technologies, National University of Singapore, 3 Science Drive 2, 117543, Singapore}

\author{Mile Gu}
\email{mgu@quantumcomplexity.org}
\affiliation{School of Physical and Mathematical Sciences, Nanyang Technological University, 637371, Singapore}
\affiliation{Complexity Institute, Nanyang Technological University, 637335, Singapore}
\affiliation{Centre for Quantum Technologies, National University of Singapore, 3 Science Drive 2, 117543, Singapore}

\maketitle

\noindent Memories are essential for information processing, from communication to sensing and computation. In the context of quantum technologies, such memories must also faithfully preserve the uniquely quantum properties that enable quantum advantages, including quantum correlations and coherent superpositions~\cite{duan2001long}. This has motivated extensive work in experimental realisations across numerous physical platforms~\cite{zhong2015optically,wang2017single}, and presents a pressing need to find operationally meaningful means to compare quantum memories across diverse physical and functional settings. In contrast, present approaches towards~{detecting and benchmarking the quantum properties of memories} are often ad-hoc, involving experimentally taxing process tomography, or only furnishing binary measures of performance based on tests of entanglement and coherence preservation~\cite{chuang1997prescription,PhysRevLett.78.390,PhysRevLett.86.4195,PhysRevA.78.032333,PhysRevA.81.060306,PhysRevA.88.042335,memoryResource18prx,PhysRevA.99.062319}.

{Our work addresses these issues by envisioning memory as a physical resource. We provide a means to quantify this resource by asking}: how much noise can a quantum memory sustain before it is unable to preserve uniquely quantum aspects of information? Defining this as the robustness of a quantum memory (RQM), we demonstrate that the {quantifier} has diverse operational relevance {in benchmarking the quantum advantages enabled by a memory} --- from speed-up in statistical sampling to nonlocal quantum games (see  Fig.~\ref{Fig:summary}). We prove that RQM behaves like a physical resource {measure,} representing the number of copies of a pure idealised qubit memory that are required to synthesise the target memory. We show the measure to be exactly computable for many relevant cases, and introduce efficient general bounds through experimental and numerical methods. The quantifier is, in particular, experimentally accessible without full tomography, enabling immediate applications in benchmarking different memory platforms and error sources, as well as providing a witness for non-Markovianity. 
We experimentally test our benchmark on the five-qubit IBM Q hardware for different types of error, demonstrating its versatility.
{In addition, the generality of our methods within the broad physical framework of quantum resource theories~\cite{horodecki2013quantumness,coecke_2016,chitambar2018quantum} ensures that many of our operational interpretations of the RQM can also extend to the study of more general quantum processes~\cite{2018arXiv181110307K,gour2019comparison,takagi2019general,liu2019operational,liu2019resource,takagi_2020}, including general resource theories of quantum channels, gate-based quantum circuits, and dynamics of many-body physics.}
Our work thus presents an operationally meaningful, accessible, and practical performance-based measure for benchmarking quantum processors that is immediately relevant in today's laboratories.

\begin{figure*}
\centering
  \includegraphics[width=0.8\linewidth]{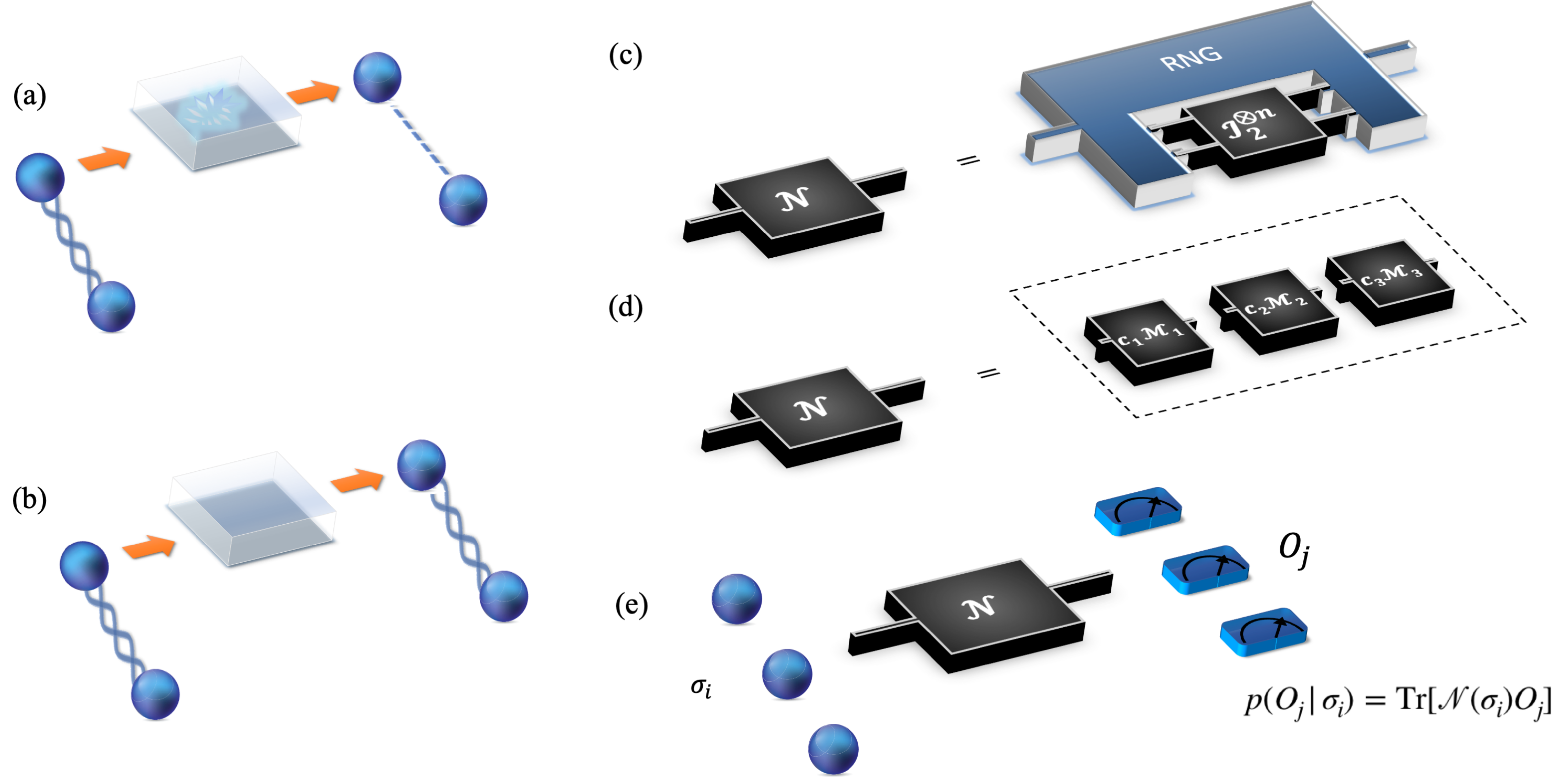}
  \caption{This work focuses on the resource theory of quantum memories. We define (a) entanglement-breaking memories as free resources and propose the RQM as a resource measure of (b) a quantum memory. We consider three operational interpretations of the measure in (c) one-shot memory synthesis with resource non-generating (RNG) transformations; (d) classical simulation of the measurement statistics of  quantum memory; and (e) a family of two-player nonlocal quantum games generalising state discrimination. }\label{Fig:summary}
\end{figure*}

\vspace{.4cm}
\noindent\textbf{{Framework.}}
Any quantum memory can be viewed as a channel in time --- mapping an input state we wish to encode into a state we will eventually retrieve in the future. An ideal memory preserves all information, such that proper post-processing operations on the output state can always undo the effects of the channel. Such channels preserve all state overlaps, in the sense that any pair of distinguishable input states remain distinguishable at output. In contrast, this is not possible with classical memories that store only classical data. To distinguish orthogonal states in some basis $\ket{k}$, we are forced to measure in this basis and record only the classical measurement outcome $k$. Such a measure-and-prepare process will never distinguish $\ket{0}+\ket{1}$ from  $\ket{0}-\ket{1}$. In fact, this procedure exactly encompasses the class of all entanglement-breaking (EB) channels~\cite{horodecki2003entanglement,holevo2008entanglement}: if we store one part of an entangled bipartite state within classical memory, the output is always separable. As such, classical memories are mathematically synonymous with EB channels.

To {systematically} characterise how well a general memory preserves quantum information, 
we consider how robust it is against noise. We define \emph{robustness of quantum memories} (RQM) as the minimal amount of a classical memory that needs to be mixed with the target memory $\mc N$ such that the resultant probabilistic mixture is also classical:
\begin{equation}
\begin{aligned}
	\mc R(\mc N) &= \min_{\mc M\in \mathrm{EB}}\left\{s\ge0\bigg| \frac{\mc N+s\mc M}{1+s}\in \mathrm{EB}\right\}, \\
\end{aligned}
\end{equation}
where the minimisation is over the set of all entanglement-breaking channels $\mathrm{EB}$. We explicitly prove that the robustness measure is a bona fide resource measure of quantum memories, satisfying all necessary operational properties. 
{Crucially, we show that the robustness satisfies \emph{monotonicity}~---~a memory's RQM can never increase under any resource non-generating (RNG) transformation, that is, any physical transformation of quantum channels that maps EB channels only to EB channels. We thus refer to such transformations as \textit{free} within the resource theory of quantum memories.
} 
Commonly encountered free transformations include pre- or post-processing with an arbitrary channel or, more generally, the class of so-called classically correlated transformations~\cite{memoryResource18prx}.

\vspace{0.2cm}
\noindent \textbf{{Operational interpretations.}} We illustrate the operational relevance of RQM in three distinct settings. The first is \emph{memory synthesis}. From the perspective of physical resources, one important task is to synthesise a target resource {by expending a number of ideal resources, which can be thought of as the ``currency'' in this process. Intuitively, a more resourceful object would be harder to synthesise and hence require more ideal resources, allowing us to understand the required number of ideal memories as the resource cost of a given memory.} In entanglement theory, an analogous concept involves determining the minimum number of Bell pairs that are required to engineer a particular entangled state using free operations (entanglement cost)~\cite{Buscemi2011entanglement,Brandao2011oneshot}. \yx{For quantum memories, we consider an ideal qubit memory $\mc{I}_2$ as the identity channel that perfectly preserves any qubit state. The task of \emph{single-shot memory synthesis} is then to convert $n$ copies of ideal qubit memories $\mc{I}_2^{\otimes n}$ to the target memory $\mc{N}$ via a free transformation. We show that the robustness measure lower bounds the number $n$ of the requisite ideal memories, i.e., $n\ge \lceil\log_2( \lceil\mc R(\mc N)\rceil+1)\rceil$.} Therefore, a larger robustness indicates that the memory requires more ideal resources to synthesise.  Furthermore, we show that there always exists an optimal RNG transformation that saturates this lower bound, and thus the robustness tightly captures the optimal resource cost for this task. We summarise our first result as follows.

\vspace{0.2cm}
\noindent\textbf{Theorem 1.}
\yx{\emph{The minimal number of ideal qubit memories required to synthesise a memory $\mc N$ is $n=\lceil\log_2( \lceil\mc R(\mc N)\rceil+1)\rceil$.}}
\vspace{0.2cm}

\noindent In Methods, we further consider imperfect memory synthesis by allowing an error $\varepsilon$ and show that the optimal resource cost is characterised by a smoothed robustness measure with smoothing parameter $\varepsilon$. Theorem 1 thus corresponds to the special case of $\varepsilon=0$.

\vspace{0.2cm}
In the second task, we consider \emph{classical simulation of quantum memories}. The motivation here is analogous to computational speed-up --- the observational statistics of any quantum algorithm can be simulated on a classical computer, albeit at an exponential overhead. Similarly, one strategy for simulating quantum memories is to perform full tomography of the input state and store the resulting classical density matrix. Then, at the output of the memory, the input state $\rho$ is reconstructed and any observational statistics on $\rho$ can be directly obtained. This method clearly requires an exponential amount of input samples for an $n$-qubit memory --- and thus results in an exponential overhead in resources and speed. 

Formally, the functional behaviour of any memory is fully described by how its observational statistics vary as a function of input, i.e., the set of expectation values $\mathrm{Tr}(O \mc N (\rho))$, for each possible observable $O$ and input state $\rho$. In order to simulate a quantum channel using only classical memories, we aim to estimate this quantity to some fixed additive error with at most some fixed admissible failure probability by taking samples --- inputting $\rho$ in a classical memory, measuring $O$, and repeating to get expectation value estimates. Intuitively, the more non-classical a memory is, the more classical samples will be required to simulate its statistics effectively. We thus define the simulation overhead $C$ as the increase in the number of samples required when using only classical memories, versus having access to $\mc N$ itself. We then prove that the optimal overhead is given exactly by the RQM of the quantum memory.

\vspace{0.2cm}
\noindent\textbf{Theorem 2.}
\emph{The minimal overhead --- in terms of extra runs or input samples needed --- to simulate the observation statistics of a quantum memory $\mc N$ is given by $
	\mc C_{\min}  = (1+2\mc R(\mc N))^2$.}
\vspace{0.2cm}

\noindent \mg{For entanglement-breaking channels $\mc M$, the robustness $\mc R(\mc M)$ vanishes and hence $\mc C_{\min}(\mc M)=1$, aligning with the intuition that classical memories require no extra simulation cost. For $n$ ideal qubit memories, $\mc R ({\mc I}_2^{\otimes n}) =2^n-1$ and hence the classical simulation overhead scales exponentially with $n$}.

\vspace{0.2cm}

In the third setting, we consider the capability of quantum memories to provide advantages in a class of two-player \emph{nonlocal quantum games}. Related games of this type have previously been employed in understanding features of Bell nonlocality~\cite{buscemi_2012-1} and detecting quantum memories~\cite{memoryResource18prx}. Consider then a set of states $\{\sigma_i\}$, from which one party (Alice) selects one state uniformly at random and encodes it in a memory $\mc N$. Her counterpart Bob is given this memory and tasked with guessing which of the states $\{\sigma_i\}$ was encoded by performing a measurement $\{O_j\}$. The probability that Bob guesses $\sigma_j$ when the input state is $\sigma_i$ is given by $\tr[\mc N(\sigma_i)O_j]$. Thus, by associating with each such guess a coefficient $\alpha_{ij} \in \mathbb{R}$ we can define the payoff of the game --- this can be used to give different weights to corresponding states, or to penalise certain guesses. The performance of the two players in the game defined by $\mc G = \{\{\alpha_{ij}\}, \{\sigma_i\}, \{O_j\}\}$ is then evaluated using the average payoff function,
\begin{equation}
    \mc P(\mc N,\mc G)= \sum_{i,j} \alpha_{ij} \tr[\mc N(\sigma_i)O_j].
\end{equation}
Such games can be considered as a generalisation of the task of quantum state discrimination, as can be seen by taking $\alpha_{ij} = \delta_{ij} p_i$ for some probability distribution $\{p_i\}$.
We see that the players' maximum achievable performance is limited by Bob's capacity to discern Alice's inputs, and thus each such game serves as a gauge for the memory quality of $\mc N$. In order to establish a quantitative benchmark for the resourcefulness of a given memory, we can then compute the best advantage it can provide in the same game $\mc G$ over all classical memories. To make such a problem well-defined, we will constrain ourselves to games for which the payoff $\mc P(\mc M, \mc G)$ is non-negative. In the Methods we then show that the maximal capabilities of a quantum memory in this setting are exactly measured by the robustness.

\vspace{0.2cm}
\noindent\textbf{Theorem 3.}
\emph{The advantage that a quantum memory $\mc N$ can provide over classical memories in all nonlocal quantum games is given by
\begin{equation}
\max_{\mc G} \frac{P(\mc N, \mc G)}{\max_{\mc M \in \mathrm{EB}}\mc P (\mc M, \mc G)} = \mc R(\mc N) + 1.
\end{equation}}

\noindent 
We will shortly see that such games, in addition to showcasing another operational aspect of the robustness, allow us to efficiently bound $\mc R(\mc N)$ in many relevant cases.

\begin{figure}[t]
\centering
  \includegraphics[width=0.9\linewidth]{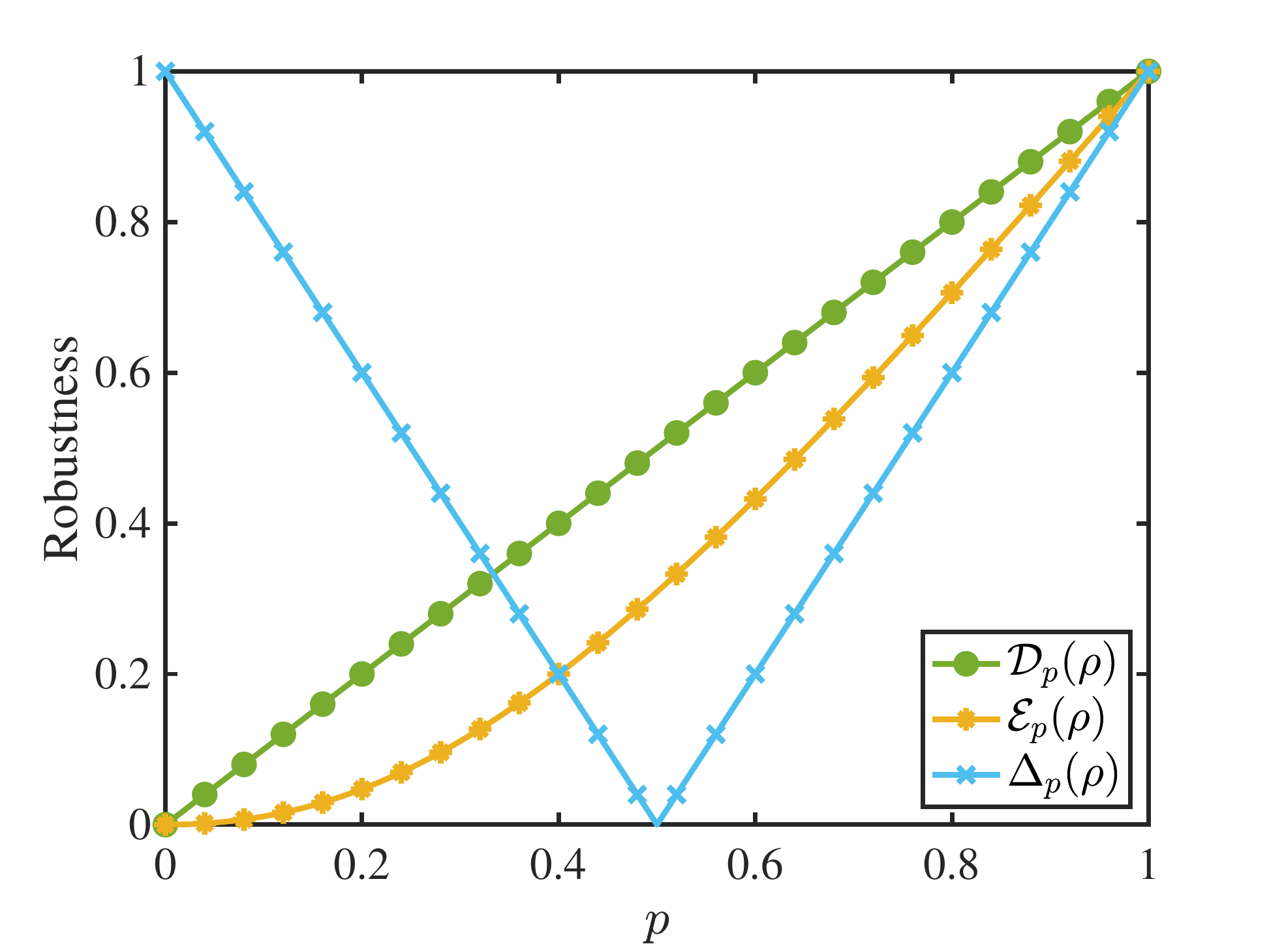}\\
  (a)\\
\includegraphics[width=0.9\linewidth]{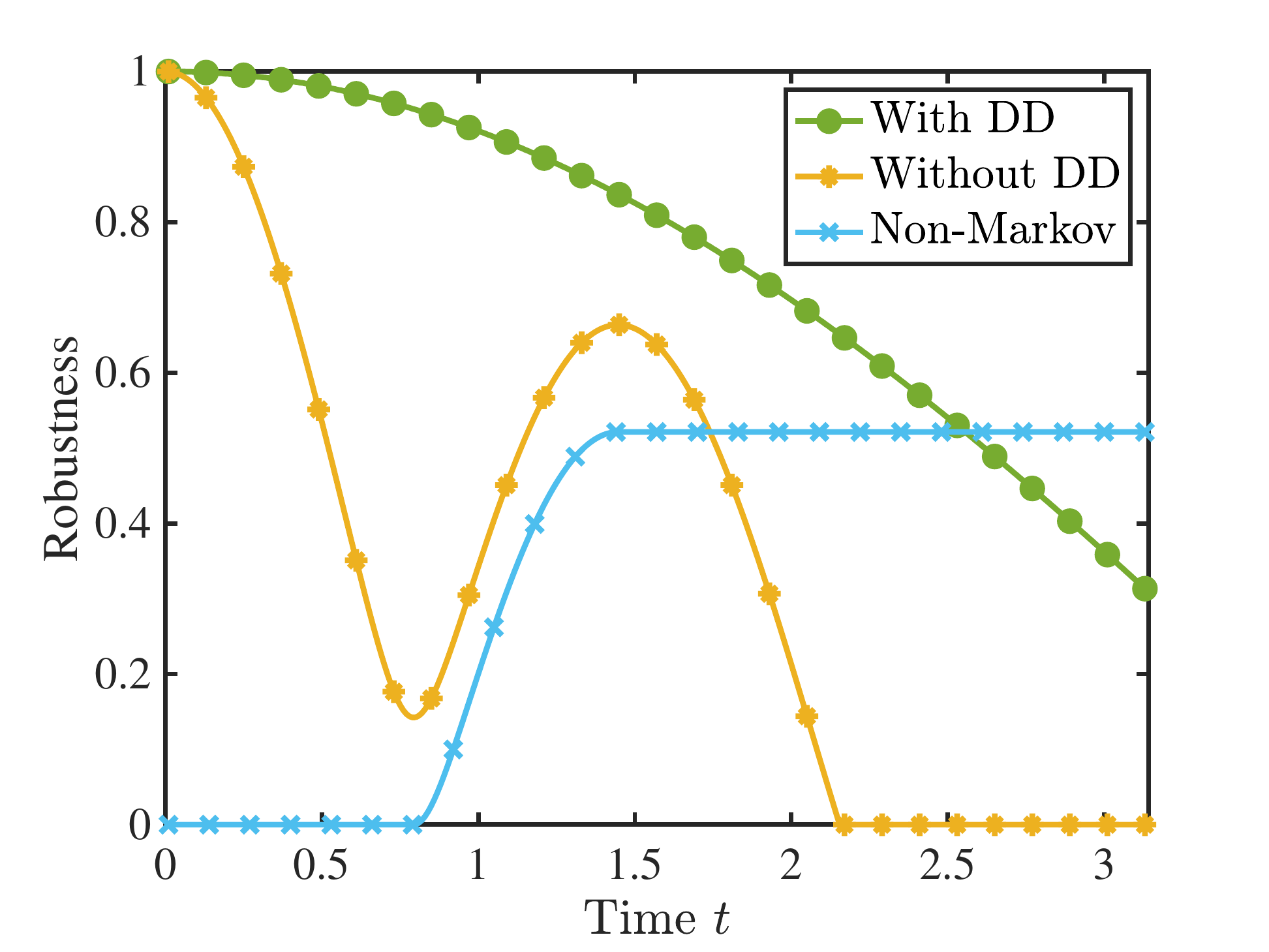}\\
    (b)
    \caption{
  Numerical evaluation of the exact value of robustness of quantum memories.
  (a) Robustness of memories with qubit inputs and computational basis $\{\ket{0},\ket{1}\}$ for dephasing channels $\Delta_p(\rho)=p\rho+(1-p)Z\rho Z$, stochastic damping channels $\mathcal D_p(\rho) = p\rho+(1-p)|{0}\rangle\langle{0}|$, and erasure channels $\mathcal E_p(\rho) = p\rho+(1-p)|{2}\rangle\langle{2}|$ with $\ket{2}$ orthogonal to $\{\ket{0},\ket{1}\}$. 
  {(b) Memory robustness under dynamical decoupling (DD) and its quantification of non-Markovianity. We consider a qubit memory (M) coupled to a qubit bath (B) with an initial state $\rho_B(0)=0.4\ketbra{0}{0}+0.6\ketbra{1}{1}$ and an interaction Hamiltonian $H= 0.2(X_M\otimes X_B+Y_M\otimes Y_B)+Z_M\otimes Z_B$. Here $X, Y, Z$ are the Pauli matrices. We consider the evolution with time $t$ from $0$ to $\pi$. To decouple the interaction, we apply $X$ operations on the memory at a constant rate. We show that the memory robustness can be  enhanced via dynamical decoupling (DD). Furthermore, as the memory robustness can increase with time, we calculate the non-Markovianity using the robustness derived measure as defined in Eq.~\eqref{Eq:nonMark}.}
  }\label{Fig:simulation}
\end{figure}

\vspace{0.2cm}

\noindent \textbf{{Computability and measurability.}}
We can efficiently detect and bound the robustness of a memory through the performance of the memory in game scenarios. Specifically, consider games $\mc G$ such that all classical memories achieve a pay-off in the range $[0,1]$. By Thm.~3 we know that any such game $\mc G$ provides a lower bound $\mc P(\mc N, \mc G)-1$ on $\mc R(\mc N)$, akin to an entanglement witness quantitatively bounding measures of entanglement~\cite{eisert_2007}. This provides a physically accessible way of bounding the robustness measure by performing  measurements on a chosen ensemble of states, and in particular there always exists a choice of a quantum game $\mc G$ such that $\mc P(\mc N, \mc G)-1$ is exactly equal to $\mc R(\mc N)$. This approach makes the measure accessible also in experimental settings, avoiding costly full process tomography. We use this method to explicitly compute the robustness of some typical quantum memories in Fig.~\ref{Fig:simulation}(a), with detailed construction of the quantum games deferred to the Supplemental Materials.

\mg{In addition to the above linear witness method}, we also give non-linear witnesses of a memory $\mc N$ based on the moments of its Choi state. Consider channels $\mc N$ with input dimension $d$ and $k=0,1,\dots,\infty$, \mg{in Supplementary Materials}, we prove
$
    \mc R(\mc N)\ge d^{\frac{k-1}{k}}\left(\tr\left[\left(\Phi_{\mc N}\right)^k\right]\right)^\frac{1}{k}-1
$,
where $\Phi_{\mc N}$ is the Choi state of $\mc N$. Higher values of $k$ provide tighter lower bounds, which can be measured in experiment by implementing a generalised swap test on {$k$ copies of the channel. In the limit $k\rightarrow\infty$} we obtain the strongest bound, which depends only on the maximal eigenvalue of the Choi state. Remarkably, the bound is actually tight for all qubit-to-qubit and qutrit-to-qubit channels.

\vspace{0.2cm}

\noindent\textbf{Theorem 4.}
\emph{The RQM of any quantum memory $\mc N$ with input dimension $d_A$ and output dimension $d_B$ can be lower bounded by 
\begin{equation}\label{eq:rob_bound_eig}
    \mc R(\mc N)\ge d_A \max\operatorname{eig}(\Phi_{\mc N}) - 1,
\end{equation}
and equality holds when $d_A \leq 3$ and $d_B = 2$.}

\vspace{0.2cm}

\noindent We stress that this provides an exact and easily computable expression for the robustness for low-dimensional channels. This contrasts with related measures of entanglement of quantum states such as the robustness of entanglement~\cite{vidal1999robustness}, for which no general expression exists even in $2\times 2$-dimensional systems.

Given a full description of the memory, we can also provide efficiently computable numerical bounds on the robustness via a semi-definite program, which we show to be tight in many relevant cases. We leave the detailed discussion to Supplementary Materials.

\vspace{0.2cm}
\noindent \textbf{{Applications.}} \mg{The robustness of quantum memories, being information theoretical in nature, applies across all physical and operational settings. This enables its immediate applicability to many present studies of quantum memory. For example, non-Markovianity \jt{and mitigation of errors resulting from non-Markovianity are widely studied problems} in the context of quantum memories. RQM can be used both to identify the former, and measure the efficacy of the latter.}

In particular, considering a memory $N_t$ that stores states from time $0$ to $t\ge 0$, we can quantify its non-Markovianity as 
\begin{equation}\label{Eq:nonMark}
\mc I(T)=\int_{0}^{T} dt\max\left\{0,\frac{d\mc R(\mc N_t)}{dt}\right\}.
\end{equation}  
For any Markovian process $\mc N_t$, the robustness measure $\mc R(\mc N_t)$ is a decreasing function of time owing to monotonicity of $\mc R$ (see Methods). Thus $\mc I(T)=0$ for any Markovian process $\mc N_t$, and  nonzero values of $\mc I(T)$ directly quantify the memory's non-Markovianity in a similar way to Ref.~\cite{PhysRevLett.105.050403}.
\mg{Meanwhile, the goal of any error-mitigation procedure is to preserve encoded qubits. Thus, the characterisation of an increase in the RQM of relevant encoded sub-spaces provides a universal measure of the efficacy for any such behaviour.}

\mg{In Fig.~\ref{Fig:simulation}(b), we illustrate these ideas using a single-qubit memory subject to unwanted coupling from a qubit bath. The robustness of quantum memories degrades over time (yellow-starred line) - but has a revival around $t = 1$, indicating non-Markovianity. Indeed, plotting $\mc I(t)$, \jt{we see clear signatures of non-Markovian effects} arise at this moment (cyan-crossed line). Meanwhile, the green-dotted line quantifies how dynamical decoupling  improves this memory through increased RQM. This improvement has a direct operational interpretation. {For example, the approximate 4-fold increase in robustness around $t = 0.8$ indicates that a quantum protocol that runs on a dynamically decoupled quantum memory could be much harder to simulate than its counterpart.}}

\begin{figure*}[t]
\centering
\includegraphics[width=1\linewidth]{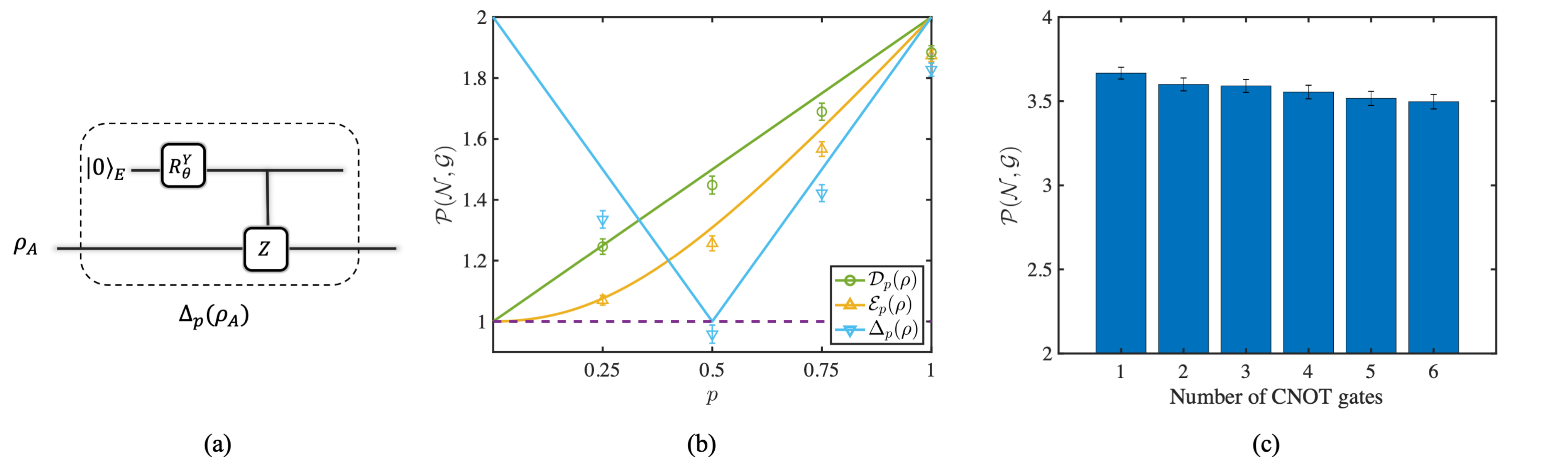}
    \caption{
    Experimental verification of the benchmark with IBM Q hardware.  (a) Circuit diagram for realising the dephasing channel.
  (b) The RQM of dephasing channels $\Delta_p(\rho)=p\rho+(1-p)Z\rho Z$, stochastic damping channels $\mathcal D_p(\rho) = p\rho+(1-p)|{0}\rangle\langle{0}|$, and erasure channels $\mathcal E_p(\rho) = p\rho+(1-p)|{2}\rangle\langle{2}|$ with $\ket{2}$ orthogonal to the basis $\{\ket{0},\ket{1}\}$. We synthesise the noise channels by interacting the target system with up to two ancillary qubits. We measure the payoff of quantum games $\mc P(\mc N,\mc G)$ which lower bounds the RQM as $\mc R(\mc N)\ge \mc P(\mc N,\mc G)-1$.   (c) Benchmarking IBM Q hardware via the RQM of sequential controlled-X (CX) gates. We interchange the control and target qubit so that two sequential CX gates will not cancel out. For example, denote $CX^0_1$ to be the CX gate with control qubit 0 and target qubit 1; the three CX gates is the swap gate $CX^0_1CX^1_0CX^0_1\equiv SWAP$ and the six controlled-X gates becomes the identity gate $CX^1_0CX^0_1CX^1_0CX^0_1CX^1_0CX^0_1\equiv I_4$.  The error bar is three times the standard deviation for both plots. 
  }\label{Fig:experiment}
\end{figure*}

\vspace{0.2cm}
\noindent \textbf{{Experiment.}} We experimentally verify our benchmarking method on the `ibmq-ourense' processor on the IBM Q cloud. 
We first consider a proof-of-principle verification of the scheme by estimating the RQM of three types of single-qubit noise channels~---~the dephasing channel, stochastic damping channel, and erasure channels. We synthesise the noise channels by entangling the target state with ancillary qubits. For example, the dephasing channel $\Delta_p(\rho_A)=p\rho_A+(1-p)Z\rho_A Z$ can be realised by the circuit in the dashed box of Fig.~\ref{Fig:experiment}(a), 
where we input an ancillary state $\ket{0}_E$, rotate it with $R^Y_{\theta}=\exp(-i\theta Y/2)$, and apply a controlled-Z. Here $\theta=2\arccos(\sqrt{p})$ and $Y$ is Pauli-Y matrix. We exploit the quantum game  approach to estimate the RQM of the three types of noise channels. We choose a normalised quantum game $\mc G$ with the maximal payoff for EB channels of $\max_{\mc M \in \mathrm{EB}} \mc P(\mc M ,\mc G)=1$, so that the robustness of memory $\mc N$ can be lower bounded by  $\mc R(\mc N)\ge \mc P(\mc N,\mc G)-1$. For each input-output setting $(\sigma_i,O_j)$, we measure the probability $p(j|i)=\tr[O_j\mc N(\sigma_i)]$ with 8192  experimental runs. The payoff is obtained as a linear combination of the probabilities $\mc P(\mc N,\mc G)=\sum_{i,j}\alpha_{i,j}p(j|i)$ with real coefficients $\alpha_{i,j}$.
As shown  in Fig.~\ref{Fig:experiment}(b), the experimental data (circles, upper and lower triangles) aligns well with the theoretical result (solid lines), with a deviation of less than $0.13$. The deviation mostly results from the inherent noise in the hardware, especially the notable two-qubit gate error and the readout error.

Next, we show that the RQM can be applied to benchmark quantum gates and quantum circuits.
Conventional quantum process benchmarking approaches, such as randomised benchmarking~\cite{emerson2005scalable,PhysRevA.77.012307}, generally focus on characterising the similarity between the noisy circuit and the target circuit. In contrast, our method is concerned with the capability of the noisy quantum processor in preserving quantum information, which can be thus regarded as an alternative operational approach for  benchmarking processes.
In the experiment, we focus on the two-qubit controlled-X (CX)  gate, a standard gate used for entangling qubits. We sequentially apply $n$ (up to six) CX gates with interchanged control and target qubits for two adjacent gates. For example, one, three, and six CX gates correspond to the CX gate, the swap gate, and the identity gate, respectively.

Assuming that the dominant error is due to depolarising or dephasing effects, we estimate the RQM of each circuit via the correspondingly designed quantum game. 
As shown in Fig.~\ref{Fig:experiment}(c), we can see that although the robustness with one CX is $2.667\pm 0.106$, it only slowly decreases to $2.497 \pm 0.115$ for six CX gates. {Our results thus  indicate that while the CX gate is imperfect (with an average 0.0340 decrease of robustness for each CX gate), the dominant noise of the two qubit circuit may instead stem from imperfect state preparation and measurement (roughly leading to a 0.3 decrease in robustness). We also note that the large robustness loss of a single CX gate might also be due to the existence of other errors, which would imply that the choice of the quantum game could be further optimised.  However, whenever the quantum game gives a large lower bound for the robustness, this is sufficient to ensure that the quantum process performs well in preserving quantum information. To demonstrate this, we consider the circuit ${CX}^0_2\cdot {CX}^0_1$ for preparing the three-qubit GHZ state. We lower bound the robustness as $5.837\pm 0.548$, verifying that the three-qubit noisy circuit can preserve more quantum information than all two-qubit circuits, whose robustness is upper bounded by 3.  }We leave the detailed experimental results and analysis to the Supplementary Materials.

\vspace{0.4cm}
\noindent \textbf{{Discussion.}}
\noindent {In this work, we introduced an operationally meaningful, practically measurable, and platform-independent benchmarking method for quantum memories.   
We defined the robustness of quantum memories and showed it to be an operational measure of the quality of a memory in three different practical settings.} \mg{The greater the robustness of a memory, the more ideal qubit memories are needed to synthesise the memory; the more classical resources are required to simulate its observational statistics; and the better the memory is at two-player nonlocal quantum games based on state discrimination. The measure can be evaluated exactly in low-dimensional systems, and efficiently approximated both numerically by semi-definite programming and experimentally through measuring suitable observables.} This thus constitutes a promising means to quantify the quantum mechanical aspects of information storage, and provides practical tools for benchmarking quantum memories across different experimental platforms and operational settings. {The theory is applicable across different physical platforms exhibiting any known type of error source, as we experimentally confirm on the five-qubit IBM Q hardware. With the development of near-term noisy intermediate-scale quantum technologies~\cite{preskill2018quantum,arute2019quantum}, we anticipate that our quantifier can become an industry standard for benchmarking quantum devices. }

{From a theoretical perspective, our work also constitutes a significant development in the resource theory of quantum memories.} The only previously known general measure of this resource involved a performance optimisation over a large class of possible quantum games~\cite{memoryResource18prx}, thus making it difficult to evaluate, experimentally inaccessible, and obscuring a direct quantitative connection to tasks of practical relevance --- the robustness explicitly addresses all of these issues.
Furthermore, the generality of the resource-theoretic framework ensures that the tools developed here for quantum memories can be naturally extended to other settings, including purity, coherence, entanglement of channels~\cite{BHLS03,Kaur_2017,liu2019operational,theurer2018quantifying,2019arXiv190702552G,2019arXiv190704181B, yuan2019hypothesis,gour2018entropy}, and the magic of operations~\cite{seddon2019quantifying, wang2019quantifying}, etc.
Another direction is to consider infinite-dimensional quantum systems, such as the optical modes of light. 
Finally, memories are essentially a question of reversibility, and thus have a natural connection to heat dissipation in thermodynamics~\cite{goold2016role,berut2012experimental}. Indeed, recent results show connections between free energy and information encoding~\cite{narasimhachar2019quantifying}, and thus present a natural direction towards understanding what thermodynamic consequences quantum memory quantifiers may have.

\vspace{0.2cm}
\noindent \textbf{Methods}

{\small
\noindent Here we present properties of the robustness measure, formal statements of Theorems 1-4 and sketch their proofs. Full version of the proofs and details on the numerical simulations can be found in Supplementary Materials.\\

\noindent \textbf{Properties of RQM.}
Recall the definition of RQM
\begin{equation}
    	\mc R(\mc N) = \min\left\{s\ge 0: \exists \mc M\in \mathrm{EB}, \, \textrm{s.t.} \, \frac{\mc N+s\mc M}{1+s}\in \mathrm{EB} \right\},
\end{equation}
where our chosen set of free channels are the entanglement-breaking (EB) channels. Define \emph{free transformation} $\mc O$ as the set of physical transformations on quantum channels (super-channels) that map EB channels to EB channels, i.e. $\mc O=\{\Lambda:\Lambda(\mc M)\in \mathrm{EB},\forall\mc M\in \mathrm{EB}\}$. This class includes, for instance, the family of classically correlated transformations, which were considered in~\cite{memoryResource18prx} as a physically-motivated class of free transformations under which quantum memories can be manipulated. In particular, transformations $\Lambda(\mc N)=\mc M_1\circ\mc N\circ\mc M_2$ with arbitrary pre- and post-processing channels $\mc M_1,\mc M_2$ are free. We show that RQM satisfies the following properties.

\vspace{0.15cm}

\noindent\emph{Non-negativity.} $\mc R(\mc N)\geq 0$ with equality if and only if $\mc N\in \mathrm{EB}$.\\
\emph{Monotonicity.} $\mc R$ does not increase under any free transformation, $\mc R(\Lambda(\mc N))\leq \mc R(\mc N)$ for arbitrary $\mc N$ and $\Lambda\in\mc O$.\\
\emph{Convexity.}~$\mc R$ does not increase by mixing channels, $\mc R\left(\sum_i p_i\mc N_i\right)\leq \sum_i p_i \mc R(\mc N_i)$.

\vspace{0.15cm}

\noindent Additional properties such as bounds under tensor product of channels are presented in Supplementary Materials. 

\vspace{0.15cm}

\noindent\emph{Proof of properties.} Non-negativity follows directly from the definition. For monotonicity, suppose $s=\mc R(\mc N)$ with the minimisation achieved by $\mc M$ such that
\begin{equation}
		\frac{1}{s+1}\mc N+\frac{s}{s+1}\mc M=\mc{M}'\in \mathrm{EB}.    
\end{equation}
Apply an arbitrary free transformation $\Lambda$ on both sides and using linearity, we obtain $\frac{1}{s+1}\Lambda(\mc N)+\frac{s}{s+1}\Lambda(\mc M)=\Lambda\left(\mc{M}'\right)\in \mathrm{EB}$. Therefore by definition $\mc R(\Lambda(\mc N))\leq s=\mc R(\mc N)$. For convexity, suppose $s_i=\mc R(\mc N_i)$ with the minimisation achieved by $\mc M_i$ for each $i$ and let $\mc M_i'=(\mc N_i+s_i\mc M_i)/(1+s_i)$. Let $s =\sum_i p_i s_i$, $\mc N = \sum_i p_i\mc N_i$ and $\mc M =\frac{1}{s}\sum_i p_i s_i\mc M_i$, then by convexity of the set of EB channels
\begin{equation}
\begin{aligned}
	\frac{1}{s+1}\mc N+\frac{s}{s+1}\mc M &= 
	\frac{1}{s+1}\sum_ip_i(s_i + 1)\mc M_i'\in \mathrm{EB},
\end{aligned}			
\end{equation}
therefore by definition we have $\mc R\left(\sum_i p_i\mc N_i\right)=\mc R(\mc N)\leq s=\sum_i p_i s_i =\sum_i p_i\mc R(\mc N_i)$.\\

\noindent \textbf{Single-shot memory synthesis.} Here we study a more general scenario, imperfect memory synthesis, which allows a small error between the synthesised memory and the target memory. The resource cost for this task is defined as the minimal dimension required for the ideal qudit memory $\mc I_d$,
\begin{equation}
	R_{\textrm{syn}}^{\varepsilon}(\mc N) =\min\left\{d:\exists\Lambda\in\mc O,
  \|\Lambda(\mc{I}_d)-\mc{N}\|_{\diamond}\leq \varepsilon\right\},
\end{equation}
where $\|\cdot\|_\diamond$ denotes the diamond norm, which describes the distance of two channels. We also include a smooth parameter $\varepsilon$ of the cost which tolerates an arbitrary amount of error in the synthesis protocol. The case with  $\varepsilon=0$ corresponds to the case with exact synthesis. \yx{When considering the ideal qubit memory $\mc I_2$ as the unit optimal resource, the minimal number of ideal qubit memories $\mc I_2^{\otimes n}$ required for memory synthesis is given by $n= \lceil\log_2(R_{\textrm{syn}}^{\varepsilon}(\mc N))\rceil$
}

Correspondingly we define a smoothed version of the robustness measure by minimising over a small neighbourhood of quantum channels,
\begin{equation}\label{eq:defsmoothrobustness}
    \mc R^\varepsilon(\mc N)=\min_{\|\mc N'-\mc N\|_\diamond\leq \varepsilon}\mc R(\mc N').
\end{equation}
We prove that the smoothed robustness measure exactly quantifies the resource cost for imperfect single-shot memory synthesis.

\vspace{0.15cm}

\noindent\emph{Formal statement of Theorem 1.} For any quantum channel $\mc N$ and any $0\leq\varepsilon<1$, the resource cost for single-shot memory synthesis satisfies
\begin{equation}
    R_{\textrm{syn}}^{\varepsilon}(\mc N) =1+\lceil\mc R^\varepsilon(\mc N)\rceil.
\end{equation}
Note that by setting $\varepsilon=0$ we recover the result for perfect memory synthesis stated in the main text.

\vspace{0.15cm}

\noindent\emph{Proof.} We start by proving $R_{\textrm{syn}}^{\varepsilon}(\mc N) \geq 1+\lceil\mc R^\varepsilon(\mc N)\rceil$. The first step is to show that the robustness of the identity channel is $\mc R(\mc I_d)=d-1$. The proof of this fact is omitted here. Next we show that the desired inequality can be proven using the monotonicity property. For an arbitrary memory synthesis protocol $\Lambda(\mc I_d)=\mc N'$ where $\|\mc N'-\mc N\|_\diamond\leq\varepsilon$, we have
\begin{equation}
    \begin{aligned}
    1+\mc R^\varepsilon(\mc N)&=1+\min_{\|\mc N'-\mc N\|_\diamond\leq \varepsilon}\mc R(\mc N')\\
    &\leq 1+\mc R(\mc N')\\
    &= 1+\mc R(\Lambda(\mc I_d))\\
    &\leq 1+\mc R(\mc I_d)\\
    &=d.
    \end{aligned}
\end{equation}
Here the second line follows by definition and the fourth line follows from monotonicity. As the above inequality holds for all memory synthesis protocols, it also holds for the optimal protocol. Also notice that dimensions are integers. Thus we derive that $R_{\textrm{syn}}^{\varepsilon}(\mc N) \geq 1+\lceil\mc R^\varepsilon(\mc N)\rceil$.

To prove the other side $R_{\textrm{syn}}^{\varepsilon}(\mc N) \leq 1+\lceil\mc R^\varepsilon(\mc N)\rceil$, suppose the  channel achieves the mimum of  Eq.~\eqref{eq:defsmoothrobustness} is $\mc N'$, and let $d_c=1+\lceil\mc R(\mc N')\rceil$. To prove the desired inequality, it suffices to show that $\exists\Lambda\in\mc O$ such that $\Lambda(\mc I_{d_c})=\mc N'$. Indeed such a $\Lambda$ is a protocol that achieves the required accuracy using resource $1+\lceil\mc R^\varepsilon(\mc N)\rceil$, thus the optimal protocol should only use less resource.

Next we explicitly construct such a free transformation $\Lambda$, which transforms a quantum channel to another channel. As there is a one-to-one correspondence between Choi states and quantum channels, we give this construction based on transformation of the Choi state:
\begin{equation}\label{dilutionchannel}
		\Lambda(\Phi_{\mc C}) = \tr\left(\phi^+\Phi_{\mc C}\right)\Phi_{\mc N'}+\tr\left((I-\phi^+)\Phi_{\mc C}\right)\Phi_{\mc M},
\end{equation}
where $\Phi$ denotes the Choi state of the subscript channel and $\phi^+$ is the maximally entangled state. In the full proof we show that $\Lambda$ is a valid physical transformation, i.e. a quantum super-channel.

As it is easy to verify that $\Lambda(\mc I_{d_c})=\mc N'$, it only remains to show that $\Lambda$ is a free transformation, which maps EB channels to EB channels. To do this, first notice that as $d_c\geq 1+\mc R(\mc N')$, there exists $\mc{M},\mc{M}'\in \mathrm{EB}$ such that
	\begin{equation}
	    \frac{1}{d_c}\mc N'+\frac{d_c-1}{d_c}\mc M = \mc M'.
	\end{equation}
Then we can rewrite Eq.~\eqref{dilutionchannel} as
	\begin{equation}
		\Lambda(\Phi_{\mc C})=q\Phi_{\mc M'}+(1-q)\Phi_{\mc M},
	\end{equation}
with $q = d_c\tr\left[\phi^+\Phi_{\mc C}\right]$. When $\mc C$ is an EB channel, $\Phi_{\mc C}$ is a separable state, and we have $0\leq q\leq 1$. Thus $\Lambda(\Phi_{\mc C})$ is a separable Choi state that corresponds to an EB channel, which means that $\Lambda$ is a free transformation and concludes the proof.\\

\noindent \textbf{Simulating observational statistics.} 
Observe that the general simulation strategy is to find a set of free memories $\{\mc M_i\}\subseteq\mathrm{EB}$ such that the target memory can be linearly expanded as $\mc N=\sum_i c_i\mc M_i,c_i\in\mathbb{R}$. By using $\mc M_i$ and measuring $\tr[O\mc M_i(\rho)]$, we can obtain the target statistics as $\tr[O\mc N(\rho)]=\sum_i c_i\tr[O\mc M_i(\rho)]$ . Thus, compared with having access to $\mc N$ and directly measuring $O$, the classical simulation introduces an extra sampling overhead with a multiplicative factor $\|c\|_1^2 = \left(\sum_i|c_i|\right)^2$. In particular, suppose we aim to estimate $\tr[O\mc N(\rho)]$ to an additive error $\varepsilon$ with failure probability $\delta$. Due to Hoeffding's inequality~\cite{PhysRevLett.115.070501}, when having access to $\mc N$ we need $T_0 \propto 1/{\varepsilon^2}\log(\delta^{-1})$ samples to achieve this estimate to desired precision, and when only having access to free resources in a specific decomposition $\mc N=\sum_i c_i\mc M_i$, we need $T \propto {\|c\|_1^2}/{\varepsilon^2}\log(\delta^{-1})$ samples. The simulation overhead is thus given by
$\|c\|_1^2\propto T/T_0$.

By minimising the simulation overhead over all possible expansions, we obtain the optimal simulation cost
\begin{equation}
    \mc C_{\min} (\mc N)=\min_{\mc N=\sum_i c_i\mc M_i}\mc C(\{c_i,\mc M_i\}).
\end{equation}
Our second result shows that this optimal cost is quantified by the robustness measure.

\vspace{0.15cm}

\noindent\emph{Formal statement of Theorem 2.} For any quantum channel $\mc N$, the optimal cost for the observational simulation of $\mc N$ using EB channels is given by
\begin{equation}
    \mc C_{\min} (\mc N)=(1+2\mc R(\mc N))^2.
\end{equation}

\vspace{0.15cm}

\noindent\emph{Proof.} For any linear expansion $\mc N=\sum_i c_i \mc M_i$, denote the positive and negative coefficients of $c_i$ by $c_i^+$ and $c_i^-$, respectively. Then we have 
\begin{equation}
    \mc N = \sum_{i:c_i\ge 0} |c_i^+| \mc M_i - \sum_{i:c_i< 0} |c_i^-| \mc M_i, \, \mc M_i\in \mathrm{EB},
\end{equation}
with $\|c\| = \sum_{i:c_i\ge 0}|c_i^+|+\sum_{i:c_i< 0}|c_i^-|$. As the channel is trace preserving, taking trace on both sides we get $\sum_{i:c_i\ge 0}|c_i^+|-\sum_{i:c_i< 0}|c_i^-|=1$. Denote $s = \sum_{i:c_i< 0}|c_i^-|$, hence with $\|c\|_1 = 2s+1$, $\mc M = \sum_{i:c_i< 0}|c_i^-| \mc M_i/s$, and $\mc M' = \sum_{i:c_i\geq 0}|c_i^+| \mc M_i/(1+s)$, we have
\begin{equation}\label{eq:expansion}
    \mc N = (s+1) \mc M' - s \mc M,
\end{equation}
where by convexity of $\mathrm{EB}$ we have $\mc M,\mc M'\in \mathrm{EB}$. Therefore finding the optimal expansion is equivalent to finding the smallest $s$ such that Eq.~\eqref{eq:expansion} holds, which by definition equals to the robustness, i.e. $s_{\min}=\mc R(\mc N)$. Then we conclude that $\mc C_{\min} (\mc N)=(1+2\mc R(\mc N))^2$.\\

\noindent \textbf{Nonlocal games.} Consider a quantum game $\mc G$ defined by the tuple $\mc G=(\{\alpha_{ij}\},\{\sigma_i\}, \{O_j\})$, where $\sigma_i$ are input states, $\{O_j\}$ is a positive observable valued measures at the output, and $\alpha_{i}\in\mathbb{R}$ are the coefficients which define the particular game. The maximal performance in the game $\mc G$ enabled by a channel $\mc N$ is quantified by the payoff function $\mc P(\mc N,\mc G)=\sum_{ij}\alpha_{ij}\tr[O_j\mc N(\sigma_i)]$. Theorem 3 establishes the connection between the advantage of a quantum channel in the game scenario over all EB channels and the robustness measure. To ensure that the optimisation problem is well-defined and bounded, we will optimise over games which give a non-negative payoff for classical memories, which include standard state discrimination tasks.

\vspace{0.15cm}

\noindent\emph{Formal statement of Theorem 3.} Let $\mc G'$ denote games such that all EB channels achieve a non-negative payoff, that is,
\begin{equation}
\mc P(\mc M, \mc G') \geq 0 \; \forall \mc M \in \mathrm{EB}.
\end{equation}
Then the maximal advantage of a quantum channel $\mc N$ over all EB channels, maximised over all such games, is given by the robustness:
\begin{equation}
    \max_{\mc G'}\frac{\mc P(\mc N,\mc G')}{\max_{\mc M\in \mathrm{EB}}\mc P(\mc M,\mc G')}=1+\mc R(\mc N).
\end{equation}

\vspace{0.15cm}

\noindent\emph{Proof.} The proof is based on duality in conic optimisation (see Ref.~\cite{takagi2019general} and references therein). First we write the robustness as an optimisation problem
\begin{equation}
\begin{aligned}
	\mc R(\mc N)+1 =& \min\tr[x_1]\\
	\text{s.t. }& x_1-x_2 = \Phi_{\mc N},\\
	&x_1,x_2\in\mathrm{cone}(\mathrm{Choi}(\mathrm{EB})),
\end{aligned}
\end{equation}
where $\Phi_{\mc N}$ is the Choi state of $\mc N$, $\mathrm{Choi}(\mathrm{EB})$ denotes the Choi states of EB channels, i.e. bipartite separable Choi states, and $\mathrm{cone}(\cdot)$ represents the unnormalised version. This can be written in standard form of conic programming, based on which we can write the dual form of this optimisation problem. The dual form can be simplified as
\begin{equation}\label{dualform}
\begin{aligned}
    \mathrm{OPT}=&\max\tr[\Phi_{\mc N} W]\\
    \text{s.t. }&W=W^\dag,\\
                &\tr[\Phi_{\mc M'} W]\in [0,1], \forall \mc M'\in \mathrm{EB}.
\end{aligned}
\end{equation}
We can verify that these primal and dual forms satisfy the condition for strong duality, therefore $\mathrm{OPT}=1+\mc R(\mc N)$, and it remains to show that $\mathrm{OPT}$ equals the maximal advantage in games.

As the constraints in the dual form~\eqref{dualform} are linear, without loss of generality we can rescale the optimisation so that we only need to consider games $\mc G'$ that satisfy
\begin{equation}
    \mc P (\mc M, \mc G') \in [0,1]
\end{equation}
for any $\mc M \in \mathrm{EB}$. We can then write
\begin{equation}
\begin{aligned}
    \mc P(\mc N,\mc G')&= \sum_{i,j}\alpha_{ij}\tr[O_j\mc N(\sigma_i)],\\
    &=d\sum_{i,j} \alpha_{ij} \tr[\Phi_{\mc N}(\sigma_i^T\otimes O_j)],\\
    &:= \tr[\Phi_{\mc N}W],
\end{aligned}
\end{equation}
where $d$ is the input dimension of $\mc N$ and $W=d\sum_{i,j} \alpha_{ij} \sigma_i^T\otimes O_j$. Using this representation, the maximal advantage can be written as an optimisation problem equivalent to \eqref{dualform}. In particular, since any Hermitian matrix can be expressed in the form of $W$ for some real coefficients $\{\alpha_{ij}\}$, any witness $W$ in \eqref{dualform} can be used to construct a corresponding game $\mc G'$, and conversely any game $\mc G'$ satisfying the optimisation constraints gives rise to a valid witness $W$ in \eqref{dualform}. We thus have
\begin{equation}
     \max_{\mc G'}\mc P(\mc N,\mc G')=1+\mc R(\mc N),
\end{equation}
concluding the proof.\\

\noindent \textbf{Computability and bounds.} It is known that the description of the set of separable states is NP-hard in the dimension of the system~\cite{gurvits_2003}, and indeed this property extends to the set of entanglement-breaking channels~\cite{gharibian_2010}, making it intractable to describe in general. Nonetheless, we can solve the problem of quantifying the RQM in relevant cases, as well as establish universally applicable bounds. As described in the main text, suitably constructing nonlocal games $\mc G$ can provide such lower bounds, which can indeed be tight. More generally, one can employ the positive partial transpose (PPT) criterion~\cite{horodecki_1996-1} to provide an efficiently computable semidefinite programming relaxation of the problem, often providing non-trivial and useful bounds on the value of the RQM. We leave a detailed discussion of these methods to the Supplemental Material. In the case of low-dimensional channels, which is of particular relevance in many near-term technological applications, we can go further than numerical bounds and establish an analytical description of the RQM.

\vspace{0.15cm}

\noindent\emph{Formal statement of Theorem 4.}
For any channel $\mc N$ with input dimension $d_A$ and output dimension $d_B$, its RQM satisfies
\begin{equation}\label{eq:methodrob_bound_eig}
    \mc R(\mc N)\ge \max\{ 0, \; d_A \max\operatorname{eig}(\Phi_{\mc N}) - 1 \},
\end{equation}
and equality holds when $d_A \leq 3$ and $d_B = 2$.

\vspace{0.15cm}

\noindent\emph{Proof.} The idea behind the proof is to employ the reduction criterion for separability~\cite{horodecki_1999,cerf_1999-1}, which can be used to show that any entanglement-breaking channel $\mc M : A \to B$ satisfies $\Phi_{\mc M} \leq \frac{1}{d_A} I_{AB}$.
Therefore, the set of channels satisfying this criterion provides a relaxation of the set of EB channels, and we can define a bound on the RQM by computing the minimal robustness with respect to this set. A suitable decomposition of a channel $\mc N$ can then be used to show that, in fact, this bound is given exactly by the larger of $d_A \max\operatorname{eig}(\Phi_{\mc N}) - 1$ and $0$. In the case of $d_A \leq 3$ and $d_B = 2$, the reduction criterion is also a sufficient condition for separability, which ensures that the robustness $\mc R(\mc N)$ matches the lower bound.\\

\noindent \textbf{Experiment details.}
The processor has five qubits with $T_1$ and $T_2$ ranging from $25\sim110\mu s$, single-qubit gate error $3.3\sim 6.5 \times 10^{-4}$, two-qubit gate error $1.0\sim 1.5 \times 10^{-2}$, and read-out error $1.9\sim 4.5\times 10^{-2}$. 
Our experiments are run on the first three qubits, which have the highest gate fidelities, and the circuits are implemented with Qiskit~\cite{aleksandrowicz2019qiskit}.

}

\vspace{0.2cm}
\noindent \textbf{Acknowledgements.}
We are grateful to Ryuji Takagi for making us aware of an error in a preliminary version of this manuscript. We acknowledge Simon Benjamin and Earl Campbell for insightful discussions. This work is supported by the EPSRC National Quantum Technology Hub in Networked Quantum Information Technology (EP/M013243/1), the National Natural Science Foundation of China Grants No.~11875173 and No.~11674193, and the National Key R\&D Program of China Grants No.~2017YFA0303900 and No.~2017YFA0304004, the National Research Foundation of Singapore Fellowship No.~NRF-NRFF2016-02 and the National Research Foundation and L'Agence Nationale de la Recherche joint Project No.~NRF2017-NRFANR004 VanQuTe, the MOE Tier 1 Grant MOE$2019$-T1-$002$-$015$, the Foundational Questions Institute (FQXi) large grant FQXi-RFP-IPW-1903. Finally we acknowledge use of the IBM Q for this work. The views expressed are those of the authors and do not reflect the official policy or position of IBM, the IBM Q team, the National Foundation of Singapore or the Ministry of Education of Singapore.

\vspace{0.2cm}

\noindent \textbf{Author contributions.} 
X.Y., Y.L., and B.R.~devised the main conceptual and proof ideas. X.Y.~and Q.Z.~performed the numerical simulation and the experiment on the IBM Q cloud.  All authors contributed to the development of the theory and the writing of the manuscript. 
\vspace{0.2cm}

\noindent \textbf{Data availability.}
The authors declare that all data supporting this study are contained within the article and its supplementary files.
\vspace{0.2cm}

\noindent \textbf{Competing interests.} The authors declare no competing financial interests.
\vspace{0.2cm}

\bibliographystyle{apsrev4-1}
\bibliography{memory}

\clearpage
\onecolumngrid

\begin{center}
    \textbf{Supplemental Information: Universal and Operational Benchmarking of Quantum Memories}
\end{center}

\section{Resource framework}
We first review the framework of resource theory of memories introduced in Ref.~\cite{memoryResource18prx}.

\subsection{Resource theories of memories}
Focusing on two chronologically ordered systems $A$ and $B$, a quantum memory is described by a channel $\mc N^{A\rightarrow B}$ that maps system $A$ to $B$, i.e.,  a completely positive trace-preserving (CPTP) linear map from $\mc D(\mc H_A)$ to $\mc D(\mc H_B)$.
Here, $\mc H$ represents the Hilbert space and $\mc D(\mc H)$ represents the set of states. 
The resource theory of quantum memories $\mathbf{C}=(\mc{F},\mc{O},\mc{R})$ is a tuple with the free memory set $\mc{F}$, free transformations $\mc O$ and resource measure $\mc R$. The resource theory provides a framework to systematically study the properties of quantum memories. In this work, we mainly focus on the capability of preserving quantum information of memories.  In this section, we introduce definitions of the free memory set $\mc{F}$ and free transformations $\mc O$, and leave the discussion of resource measures $\mc R$ in the next section.
\begin{enumerate}
\item Free memories $\mc{F}$: entanglement breaking (EB) or equivalently  measure-and-prepare channels,
\begin{equation}\label{EBdef}
    \mc N^{A\rightarrow B}(\rho^A) = \sum_i \tr[\rho^A M^A_i]\sigma^B_i.
\end{equation}
Here $\{M_i^A\}$ is a POVM satisfying $M^A_i\ge 0$ and $\sum_i M^A_i=I^A$. The reason that we choose EB channels to be free is because only classical information is stored and forwarded from system $A$ to system $B$.  The channels that have the maximal resource are isometric channels as they are reversible. 

We consider the normalised Choi state of a channel $\mc N$,
\begin{equation}
	\Phi^{AB}_{\mc N} = \mc N^{A'\rightarrow B}(\Phi^+_{AA'}),
\end{equation}
with maximally entangled state ${\Phi^+_{AA'}}=1/{d}\sum_{ij} \ket{ii}\bra{jj}$. Here $d$ is the dimension of the input system $A'$. For a linear CPTP map $\mc N$, the corresponding Choi state is a normalised quantum state satisfying $\tr_B[\Phi^{AB}_{\mc N}]=I_A/d$. Conversely, for any normalised bipartite quantum state $\Phi^{AB}$ which satisfies $\tr_B[\Phi^{AB}]=I_A/d$, there is a unique quantum channel $\mc N$ whose Choi state equals to $\Phi^{AB}$ and can be espressed as
\begin{equation}
    \mc N^{A\to B}(\rho)=d\cdot\tr_A\left[(\rho^{T}\otimes I)\Phi^{AB}\right].
\end{equation}

We have the following properties for the free memory set $\mc F$:
\begin{enumerate}

\item { The set of entanglement-breaking channels is convex. If a channel $\mc M$ is of the form
\begin{equation}
    \mc M(\rho) = \sum_i \tr[P_i \rho] \psi_i
\end{equation}
with $\psi_i$ pure and $\{P_i\}$ mutually orthogonal rank one projections, then it is an extreme element of the set of entanglement-breaking channels~\cite{horodecki2003entanglement}.}

\item A quantum channel is EB if and only if its Choi state is a separable state. Equivalently, a free channel admits a Kraus decomposition as $\mc M(\rho) = \sum_i K_i \rho K_i^\dagger$ where each $K_i$ is rank one~\cite{horodecki2003entanglement}.




\end{enumerate}

\item Free super-operations $\mc{O}$: any super-operation that only transmits classical information is free,
\begin{equation}\label{Eq:superphy}
    \Lambda(\mc N^{A\rightarrow B}) = \mc V^{BE\rightarrow B'}\circ
\mc N^{A\rightarrow B} \circ\Delta^E\circ\mc U^{A'\rightarrow AE}.
\end{equation}
Here $\mc U^{A'\rightarrow AE}$ and $\mc V^{BE\rightarrow B'}$ are arbitrary quantum channels and $\Delta^E$ is a dephasing channel that enforces system $E$ to be classical.

For mathematical simplicity, we can also consider free operations as resource non-generating super-operations $\Lambda$, which map EB channels to EB channels,
\begin{equation}\label{Eq:generalsuper}
    \mc O = \{\Lambda:\Lambda(\mc N) \in \mc F, \, \forall \mc N\in \mc F\}.
\end{equation}
Meanwhile, it is not hard to see that free super-operations are also convex.

\end{enumerate}

Quantum channels can be represented with Choi states and transformations of quantum channels, i.e., super-operations, can be similarly regarded as special linear transformations of Choi states. 
Given the Choi state $\Phi^{AB}_{\mc N}$ of channel $\mc N^{A\rightarrow B}$, a super-operation $\Lambda(\mc N)$ can be equivalently described by a linear map acting on the Choi state. That is, suppose the Choi state of $\Lambda(\mc N)$ is  $\Phi^{AB}_{\Lambda(\mc N)}$, we have
\begin{equation}
    \Phi^{A'B'}_{\Lambda(\mc N)} = \phi_{\Lambda}(\Phi^{AB}_{\mc N}),
\end{equation}
where $\phi_{\Lambda}$ can be understood as a linear map from state $\Phi^{AB}_{\mc N}$ to state $\Phi^{A'B'}_{\Lambda(\mc N)}$. 

\section{Robustness of memory}
In this section, we introduce robustness measures of memories and study their properties. We also define the generalised robustness and study its properties.
\subsection{Definition}
The robustness of memories is defined as
\begin{equation}
	\mc R(\mc N) = \min\left\{s\ge 0: \exists \mc M\in\mc F, \, \textrm{s.t.} \, \frac{1}{s+1}\mc N+\frac{s}{s+1}\mc M\in \mc F\right\},
\end{equation}
with a minimisation over all possible EB channels.
We also define the generalised robustness as
\begin{equation}
	\mc R_G(\mc N) = \min\left\{s\ge 0: \exists \mc M\in\mathrm{CPTP}, \, \textrm{s.t.} \, \frac{1}{s+1}\mc N+\frac{s}{s+1}\mc M\in \mc F\right\},
\end{equation}
with a minimisation over all channels. Note that the robustness measures can be equivalently defined based on the Choi state of channels,
\begin{equation}
	\begin{aligned}
		\mc R(\mc N)  &= \mc R(\Phi_{\mc N}^+) = \min_{\mc M\in \mc F}\left\{s\ge 0: \exists \mc M'\in \mc F, \, \textrm{s.t.} \, \frac{1}{s+1}\Phi^+_{\mc N}+\frac{s}{s+1}\Phi^+_{\mc M}=\Phi^+_{\mc M'}\right\},\\
		\mc R_G(\mc N) &= \mc R_G(\Phi_{\mc N}^+)  = \min_{\mc M\in\mathrm{CPTP}}\left\{s\ge 0: \exists \mc M'\in \mc F, \, \textrm{s.t.} \, \frac{1}{s+1}\Phi^+_{\mc N}+\frac{s}{s+1}\Phi^+_{\mc M}=\Phi^+_{\mc M'}\right\}.
	\end{aligned}
\end{equation}
Note also that related measures have appeared in general resource theories of states~\cite{vidal1999robustness,Brandao15,regula_2018,anshu_2017} as well as channels~\cite{Diaz2018usingreusing,takagi2019general,liu2019resource}, where their properties were studied.

We also define the logarithmic robustness  as
\begin{equation}
	\LR(\mc N) = \log_2(1+\mc R(\mc N)),
\end{equation}
and the smoothed logarithmic robustness as
\begin{equation}
	\LR^\varepsilon(\mc N) = \min_{\|\mc N'-\mc N\|_\diamond\le \varepsilon} \LR(\mc N').
\end{equation}

Similarly, the max-entropy of a memory can be defined as
\begin{equation}
\begin{aligned}
	D_{\max}(\mc N) &= \log_2\left(1+ \mc R_G(\mc N)\right),\\
	&=\log_2\min\left\{\lambda:\exists\mc{M}\in\mc{F},\mc{N}\leq\lambda\mc{M}\right\},
\end{aligned}
\end{equation}
and the smoothed version as
\begin{equation}
	D_{\max}^\varepsilon(\mc N) = \min_{\|\mc N'-\mc N\|_\diamond\le \varepsilon} D_{\max}(\mc N').
\end{equation}

Note that, as the smoothing is defined based on the diamond norm of channels, the smoothed measures cannot be obtained by smoothing the Choi state, i.e.,
\begin{equation}
\begin{aligned}
	\LR^\varepsilon(\mc N) &\neq \LR^\varepsilon(\Phi^+_{\mc N}) =  \min_{\|\Phi^{+'}_{\mc N}-\Phi^+_{\mc N}\|\le \varepsilon} \log_2(1+\mc R(\Phi^+_{\mc N})),\\
	D_{\max}^\varepsilon(\mc N) &\neq D_{\max}^\varepsilon(\Phi^+_{\mc N})= \min_{\|\Phi^{+'}_{\mc N}-\Phi^+_{\mc N}\|\le \varepsilon} \log_2(1+\mc R_G(\Phi^+_{\mc N})).
\end{aligned}
\end{equation}

\subsection{Properties}\label{sec:properties}
Here, we focus on properties of the robustness measures. We prove it for $\mc R(\mc N)$ and the related measures. Unless otherwise mentioned, the same proof for $\mc{R}$ also holds for $\mc{R}_G$.
\begin{enumerate}
	\item \emph{Non-negativity.} For any channel $\mc N$, we have
		\begin{equation}
		\mc{R}(\mc N)\ge0, \, \LR(\mc N)\ge0, \, \LR^\varepsilon(\mc N)\ge0.
	\end{equation}
	More specifically, we also have the following:
	\begin{itemize}
	    \item $\mc{R}(\mc{N})=0\Leftrightarrow\mc{N}\in\mc{F}$.
	    \item $\LR(\mc{N})=0\Leftrightarrow\mc{N}\in\mc{F}$.
	    \item $\LR^\varepsilon(\mc{N})=0$ for all $\mc{N}\in\mc{F}$.
	\end{itemize}
	This follows directly from the definition.
	\item \emph{Monotonicity.} For a resource non-generating super-operation $\Lambda$, we have
	\begin{equation}
		\mc{R}( \Lambda(\mc N)) \le \mc{R}(\mc N), \, \LR( \Lambda(\mc N)) \le \LR(\mc N), \, \LR^\varepsilon( \Lambda(\mc N)) \le \LR^\varepsilon(\mc N).
	\end{equation}
	\begin{proof}
		We first prove $\mc{R}( \Lambda(\mc N)) \le \mc{R}(\mc N)$. Suppose the minimisation of $s=R(\mc N)$ is achieved with channel $\mc M$, we have
		\begin{equation}
		\frac{1}{s+1}\mc N+\frac{s}{s+1}\mc M=\mc{M}'\in \mc F.    
		\end{equation}
		Applying the resource non-generating super-operation $\Phi$ to both sides of the above equation, we have 
		\begin{equation}
			\Lambda\left(\frac{1}{s+1}\mc N+\frac{s}{s+1}\mc M\right) = \frac{1}{s+1}\Lambda(\mc N)+\frac{s}{s+1}\Lambda(\mc M)=\Lambda\left(\mc{M}'\right)\in \mc F.
		\end{equation}
		Since $\Lambda(\mc M)\in \mc F$, we conclude that $\mc{R}( \Lambda(\mc N)) \leq s = \mc{R}(\mc N)$. As $\LR(\mc N) = \log_2(1+\mc R(\mc N))$, we also have $\LR( \Lambda(\mc N)) \le \LR(\mc N)$. 
		
		To prove $\LR^\varepsilon( \Lambda(\mc N)) \le \LR^\varepsilon(\mc N)$, suppose $\mc N'$ achieves the minimisation of the smooth of $\LR^\varepsilon(\mc N)$, so that $\LR^\varepsilon(\mc{N})=\LR(\mc N')$. As $\|\Lambda(\mc N')-\Lambda(\mc N)\|_\diamond\le \|\mc N'-\mc N\|_\diamond\le \varepsilon$, we have 
		\begin{equation}
			\LR^\varepsilon(\mc N) = \LR(\mc N') \ge  \LR( \Lambda(\mc N')) \ge \LR^\varepsilon( \Lambda(\mc N)).
		\end{equation}
	With monotonicity, we also have that the measures are invariant under reversible transformations.
	\end{proof}
	
	As a special case of the monotonicity property, we have that the robustness measure of sequentially connected memories is upper bounded by the minimal robustness of each memory. That is,
	\begin{equation}
	    \mc R(\mc N_k\circ \dots\circ \mc N_1) \le \min_{i=1,\dots,k} \mc R(\mc N_i).
	\end{equation}
	This also holds for the logarithmic robustness and the smoothed logarithmic robustness.
	
	\item \emph{Convexity.} For a set of memories $\{\mc N_i\}$ with probability distribution $\{p_i\}$ satisfying $p_i\geq 0$ and $\sum_i p_i=1$, the averaged resource measure cannot be increased via mixing memories, i.e., 
	\begin{equation}
	    \mc R\left(\sum_i p_i\mc N_i\right) \le \sum_i p_i\mc R(\mc N_i).
	\end{equation}

	\begin{proof}
	For all $i$, suppose the minimisation of $s_i= R(\mc N_i)$ is achieved with $\mc M_i$, that is,
		\begin{equation}
			\frac{1}{s_i+1}\mc N_i+\frac{s_i}{s_i+1}\mc M_i=\mc M_i'\in \mc F.
		\end{equation}
	Let $s =\sum_i p_i s_i$, $\mc N = \sum_i p_i\mc N_i$ and $\mc M =\frac{1}{s}\sum_i p_i s_i\mc M_i$, then
		\begin{equation}
		\begin{aligned}
			\frac{1}{s+1}\mc N+\frac{s}{s+1}\mc M &= \frac{1}{s+1}\left(\sum_i p_i \mc N_i+\sum_ip_is_i\mc M_i\right),\\
			&=\frac{1}{s+1}\sum_i p_i\left(\mc N_i + s_i\mc M_i\right)\\
			&=\frac{1}{s+1}\sum_ip_i(s_i + 1)\mc M_i'\in\mc F.
		\end{aligned}			
		\end{equation}
		Since $\mc M\in\mc F$, we have
		\begin{equation}
		    \mc R\left(\sum_i p_i\mc N_i\right)=\mc R(\mc N)\leq s=\sum_i p_i s_i =\sum_i p_i\mc R(\mc N_i).
		\end{equation}
	\end{proof}

	\item 
	\emph{Relation under tensor product.} 
	For two channels $\mc N_1$ and $\mc N_2$, we have
	\begin{equation}
	\begin{aligned}
			\max\{\mc R(\mc N_1), \mc R(\mc N_2)\}	&\leq \mc{R}(\mc N_1\otimes\mc N_2) \leq  2\mc{R}(\mc N_1)\mc{R}(\mc N_2)+\mc{R}(\mc N_1)+\mc{R}(\mc N_2),\\
			\max\{\mc R_G(\mc N_1), \mc R_G(\mc N_2)\}	&\leq \mc{R}_G(\mc N_1\otimes\mc N_2) \leq  \mc{R}_G(\mc N_1)\mc{R}_G(\mc N_2)+\mc{R}_G(\mc N_1)+\mc{R}_G(\mc N_2).
	\end{aligned}
	\end{equation}
	For the logarithmic robustness and max-entropy, we also have
	\begin{equation}
	\begin{aligned}
		\max\{\LR(\mc N_1),\LR(\mc N_2)\}	&\le \LR(\mc N_1\otimes\mc N_2) \le \LR(\mc N_1)+\LR(\mc N_2)+1,\\
		\max\{D_{\max}(\mc N_1),D_{\max}(\mc N_2)\}	&\le D_{\max}(\mc N_1\otimes\mc N_2) \le D_{\max}(\mc N_1)+D_{\max}(\mc N_2),\\
	\end{aligned}
	\end{equation}
	or the tighter relation
	\begin{equation}
	\begin{aligned}
	\log ( 1 + 2 \mc R (\mc N_1\otimes\mc N_2)) \le \log ( 1 + 2 \mc R (\mc N_1)) + \log ( 1 + 2 \mc R (\mc N_2) ).
		\end{aligned}
	\end{equation}
	
	\begin{proof}
    For $i=1,2$, suppose the minimisation of $s_i= \mc R(\mc N_i)$ is achieved with $\mc M_i$, that is,
		\begin{equation}
			\frac{1}{s_i+1}\mc N_i+\frac{s_i}{s_i+1}\mc M_i=\mc M_i'\in \mc F.
		\end{equation}	
	Then 
	\begin{equation}
	\begin{aligned}
		\mc N_1\otimes\mc N_2 &= (s_1+1)(s_2+1)\mc{M}_1'\otimes\mc{M}_2'+s_1s_2\mc{M}_1\otimes\mc{M}_2-s_1(s_2+1)\mc{M}_1\otimes\mc{M}_2'-s_2(s_1+1)\mc{M}_1'\otimes\mc{M}_2.
	\end{aligned}
	\end{equation}
    Let $s=(s_1+1)(s_2+1)+s_1s_2-1$, then $\mc{N}_1\otimes\mc{N}_2$ can be expressed as 
    \begin{equation}
        \mc{N}_1\otimes\mc{N}_2=(s+1)\mc{M}'-s\mc{M},
    \end{equation}
    where
    \begin{equation}
        \mc{M}'=\frac{(s_1+1)(s_2+1)\mc{M}_1'\otimes\mc{M}_2'+s_1s_2\mc{M}_1\otimes\mc{M}_2}{(s_1+1)(s_2+1)+s_1s_2}
    \end{equation}
 and
 \begin{equation}
     \mc{M}=\frac{s_1(s_2+1)\mc{M}_1\otimes\mc{M}_2'+s_2(s_1+1)\mc{M}_1'\otimes\mc{M}_2}{2s_1s_2+s_1+s_2}.
 \end{equation}

    In conclusion,
	\begin{equation}
		\mc{R}(\mc N_1\otimes\mc N_2) \leq s = 2\mc{R}(\mc N_1)\mc{R}(\mc N_2)+\mc{R}(\mc N_1)+\mc{R}(\mc N_2).
	\end{equation}

	For the generalised robustness, we have
	\begin{equation}
	    \frac{(\mc N_1+s_1\mc M_1)\otimes(\mc N_2+s_2\mc M_2)}{(s_1+1)(s_2+1)} =  \frac{\mc N_1\otimes \mc N_2}{(s_1+1)(s_2+1)} + \frac{s_1\mc M_1\otimes \mc N_2+s_2\mc N_1\otimes \mc M_2+s_1s_2\mc M_1\otimes\mc M_2}{(s_1+1)(s_2+1)} = \mc M_1'\otimes \mc M_2'.
	\end{equation}
	Therefore 
	\begin{equation}
	    \mc R_G(\mc N) \le (s_1+1)(s_2+1) - 1 = \mc{R}_G(\mc N_1)\mc{R}_G(\mc N_2)+\mc{R}_G(\mc N_1)+\mc{R}_G(\mc N_2).
	\end{equation}
		
		To prove the lower bound, we define the \emph{partial trace of a quantum channel} as
		\begin{equation}
		    \tr_B \mc{N}_{AB}=\tr_B \mc{N}_{AB}\left(\cdot\otimes \frac{\mc{I}_B}{d_B}\right).
		\end{equation}
		By definition~\eqref{EBdef}, the partial trace of a bipartite EB channel is also an EB channel. 
		
		Suppose the minimisation of $s= \mc R(\mc N_1^A\otimes\mc N_2^B)$ is achieved with $\mc M^{AB}$, that is,
		\begin{equation}
			\frac{1}{s+1} \mc N_1^A\otimes\mc N_2^B+\frac{s}{s+1} \mc M^{AB}=\tilde{\mc{M}}^{AB}\in \mc F,
		\end{equation}	
		where superscript denotes input systems. Now, take partial trace on system B, we get
		\begin{equation}
		    \frac{1}{s+1} \mc N_1^A+\frac{s}{s+1} \tr_B\mc M^{AB}=\tr_B\tilde{\mc{M}}^{AB}\in \mc F,
		\end{equation}
		thus $\mc{R}(\mc{N}_1)\leq s$. Symmetrically, we also have $\mc{R}(\mc{N}_2)\leq s$, so the lower bound $\max\{\mc R(\mc N_1), \mc R(\mc N_2)\}	\leq \mc{R}(\mc N_1\otimes\mc N_2)$ is obtained.
		
	\end{proof}
	
	\item \emph{Stability.} For any channel $\mc N$ and any EB channel $\mc M$,
	\begin{equation}
		\mc R(\mc N\otimes \mc M)= \mc R(\mc N), \, \LR(\mc N\otimes \mc M)=\LR(\mc N).
	\end{equation}
	
	\begin{proof}
		As appending a free channel is a resource non-generating operation, we have $\mc R(\mc N\otimes \mc M)\leq \mc R(\mc N)$. Then, from property 4 we also have $\mc R(\mc N\otimes \mc M)\geq \mc R(\mc N)$, thus the equality is obtained.
	\end{proof}
\end{enumerate}

\section{Single-shot memory dilution}
\subsection{Preliminary}
Before proving the main results, we first obtain some preliminary Lemmas which are useful for the following discussions.
\begin{lemma}\label{LRidentity}
The logarithmic robustness of an ideal $d$-dimensional quantum memory $\mc{I}_d$ is
\begin{equation}
    \LR(\mc I_d)=\log_2 d.
\end{equation}
\end{lemma}
\begin{proof}
    It is equivalent to show that $\mc R(\mc I_d)=d-1$. We first prove the upper bound $\mc R(\mc I_d)\leq d-1$. Notice that the Choi state of $\mc I_d$ is the maximally entangled state $\ket{\Phi^+_{AB}}=\frac{1}{\sqrt{d}}\sum_{i=1}^d\ket{i}_A\otimes\ket{i}_B$. Using the same technique as in~\cite{vidal1999robustness}, we show that $\Phi^+_{AB}$ can be expressed as separable states
    \begin{equation}
        \Phi^+_{AB}=d\rho^+ - (d-1)\rho^-,
    \end{equation}
    where $\rho^-=\frac{1}{d(d-1)}\sum_{i\neq j}\ket{ij}_{AB}\bra{ij}_{AB}$ and $\rho^+=\frac{\Phi^+_{AB}+(d-1)\rho^-}{d}$. It is shown in~\cite{vidal1999robustness} that these are both separable states. It remains to show that these correspond to Choi states of EB channels. This is true since
    \begin{equation}
    \begin{aligned}
        \tr_B\rho^-&=\frac{1}{d(d-1)}\sum_k\sum_{i\neq j}\ket{i}_A\bra{k}\ket{j}_B\bra{i}_A\bra{j}\ket{k}_B\\
        &=\frac{1}{d(d-1)}\sum_i\sum_{j\neq i}\ket{i}\bra{i}_A\\
        &=\frac{1}{d}\sum_i\ket{i}\bra{i}_A\\
        &=\frac{\mc I_A}{d}.
    \end{aligned}
    \end{equation}
    By definition of robustness of memory, we have $\mc R(\mc I_d)\leq d-1$. Also, the lower bound $\mc R(\mc I_d)\geq d-1$ is trivial, because if $\mc R(\mc I_d)< d-1$, then using the same construction, we can show that the robustness of entanglement of $\Phi^+_{AB}$ is less than $d-1$, violating the known result in~\cite{vidal1999robustness}.
\end{proof}

\begin{lemma}\label{validsuperchannel}
For any pure state $\ket{\psi}$ and quantum channels $\mc{N}_1,\mc{N}_2$ in the same space, the linear map
\begin{equation}
    \Lambda(\Phi_\mc{C})=\tr\left(\psi\Phi_\mc{C}\right)\Phi_{\mc N_1}+\left(1-\tr\left(\psi\Phi_\mc{C}\right)\right)\Phi_{\mc N_2}
\end{equation}
is a quantum super-channel represented by operation on Choi states. Here, $\mc C$ is an input quantum channel, $\Phi_{\mc C}$ is its Choi state, and $\Lambda(\Phi_\mc{C})$ is the Choi state of the output channel.
\end{lemma}
\begin{proof}
    For an input quantum state $\omega$, the output channel acts on $\omega$ as 
    \begin{equation}
        \Lambda(\mc C)(\omega)=\tr\left(\psi\Phi_\mc{C}\right)\mc N_1(\omega)+\left(1-\tr\left(\psi\Phi_\mc{C}\right)\right)\mc N_2(\omega).
    \end{equation}
    
    In order to show that $\Lambda$ is a quantum super-channel, we only need to show that $\Lambda(\mc C)$ can be decomposed into three steps: pre-processing, action with ancillary system, and post-processing, which are constructed as follows:
    \begin{itemize}
        \item pre-processing: the input state is appended with a maximally entangled state $$\omega\to\Phi_{AB}^+\otimes \omega_C$$
        \item action with ancillary system: the channel $\mc C$ acts on system B while identity acts on systems A and C $$\Phi_{AB}^+\otimes \omega_C\to\mc I_A\otimes \mc C (\Phi_{AB}^+)\otimes \omega_C=\Phi_{\mc C}\otimes\omega_C$$
        \item post-processing: suppose the channels $\mc{N}_1,\mc{N}_2$ have Kraus operators $\{K_j\}_j,\{T_k\}_k$, respectively. We construct a post-processing channel with Kraus operators
        $$\left\{\bra{i}\ket{\psi}\bra{\psi}\otimes K_j,\bra{i}\left(I-\ket{\psi}\bra{\psi}\right)\otimes T_k\right\}_{ijk}.$$
        To see that this is a valid quantum channel, we have
        \begin{equation}
            \begin{aligned}
            &\sum_{ij}\ket{\psi}\bra{\psi}\ket{i}\bra{i}\ket{\psi}\bra{\psi}\otimes K_j^\dag K_j+\sum_{ik}\left(I-\ket{\psi}\bra{\psi}\right)\ket{i}\bra{i}\left(I-\ket{\psi}\bra{\psi}\right)\otimes T_k^\dag T_k\\
            &=\ket{\psi}\bra{\psi}\otimes I+\left(I-\ket{\psi}\bra{\psi}\right)\left(I-\ket{\psi}\bra{\psi}\right)\otimes I\\
            &=I\otimes I.
            \end{aligned}
        \end{equation}
    
    Also, we can see that the output of this channel is $$\Phi_{\mc C}\otimes\omega\to\tr\left(\psi\Phi_\mc{C}\right)\mc N_1(\omega)+\left(1-\tr\left(\psi\Phi_\mc{C}\right)\right)\mc N_2(\omega).$$
    \end{itemize}
\end{proof}

\subsection{Single-shot memory dilution}
We consider the problem of single-shot dilution or channel simulation under resource non-generating super-operations, which is defined as
\begin{equation}
	R_c^{1,\varepsilon}(\mc N) =\min_{\Lambda\in\mc{O}}\big\{\log_2d:\\
  \|\Lambda(\mc{I}_d)-\mc{N}\|_{\diamond}\leq \varepsilon\big\}.
\end{equation}
Note that compared to the definition in Methods of the main text, we used the logarithm of the dimension to represent the number of qubit memories required for the task.

\begin{theorem}\label{oneshotdilution}
For any quantum channel $\mc{N}$, its single-shot memory dilution rate is
\begin{equation}
	\LR^\varepsilon(\mc N) \le R_c^{1,\varepsilon}(\mc N) \le \LR^\varepsilon(\mc N)+1.
\end{equation}
\end{theorem}

\begin{proof}
First we prove the left hand side $\LR^\varepsilon(\mc N) \leq R_c^{1,\varepsilon}(\mc N)$. Suppose there exists a memory dilution protocol $\Lambda$ such that
\begin{equation}
    \|\Lambda(\mc{I}_d)-\mc{N}\|_{\diamond}\leq \varepsilon.
\end{equation}
With the monotonicity of the logarithmic robustness, we have
	\begin{equation}
		\begin{aligned}
			\LR^\varepsilon(\mc N) &\le \LR(\Lambda(\mc I_d)),\\
			&\le \LR(\mc I_d),\\
			&= \log_2d.
		\end{aligned}
	\end{equation}
Here, the last line follows from Lemma~\ref{LRidentity}. Since the above equation holds for any dilution protocol, we conclude that $\LR^\varepsilon(\mc N) \leq R_c^{1,\varepsilon}(\mc N)$.

    Next, we prove the right hand side $R_c^{1,\varepsilon}(\mc N) \le \LR^\varepsilon(\mc N)+1$.
	
	Suppose the optimisation in $\LR^\varepsilon(\mc N)$ is achieved with $\mc N'$, such that $\LR^\varepsilon(\mc N)=\LR(\mc N')=\log_2(1+\mc R(\mc N'))$. Denote $d_c = \lceil R(\mc N')\rceil + 1$. Then, there exists $\mc{M},\mc{M}'\in\mc F$ such that
	\begin{equation}
	    \frac{1}{d_c}\mc N'+\frac{d_c-1}{d_c}\mc M = \mc M'.
	\end{equation}
	We construct a linear map as
	\begin{equation}\label{dilutionchannel}
		\Lambda(\Phi_{\mc C}^{CD}) = \tr\left(\Phi^+_{CD}\Phi_{\mc C}^{CD}\right)\Phi_{\mc N'}^{AB}+\tr\left((I_{CD}-\Phi^+_{CD})\Phi_{\mc C}^{CD}\right)\Phi_{\mc M}^{AB},
	\end{equation}
	where systems $C$ and $D$ have dimension $d_c$. By Lemma~\ref{validsuperchannel}, we know that $\Lambda$ is a super-channel.

	Next, we verify that $\Lambda$ is a free resource non-generating super-operation. We first rewrite Eq.~\eqref{dilutionchannel} as
	\begin{equation}
		\Lambda(\Phi_{\mc C}^{CD}) = q\left(\frac{\Phi_{\mc N'}^{AB}+(d_c-1)\Phi_{\mc M}^{AB}}{d_c}\right)+(1-q)\Phi_{\mc M}^{AB}=q\Phi_{\mc M'}^{AB}+(1-q)\Phi_{\mc M}^{AB},
	\end{equation}
	with $q = d_c\tr\left(\Phi^+_{CD}\Phi_{\mc C}^{CD}\right)$. When $\mc C$ is an EB channel, $\Phi_{\mc C}^{CD}$ is a separable state, and we have $0\leq q\leq 1$. Thus $\Lambda(\Phi_{\mc C}^{CD})$ is a separable Choi state that corresponds to an EB channel.
	
	Lastly, we verify that when inputting the identity channel, the output channel $\Lambda(\mc I^{C\rightarrow D}) = \mc N'^{A\rightarrow B}$ is $\varepsilon$-close to the target channel $\mc N$. Therefore we have 
	\begin{equation}
	\begin{aligned}
		R_c^{1,\varepsilon}(\mc N) &\leq \log_2 d_c\\
		&=\log_2\left(1+\lceil R(\mc N')\rceil\right)\\
		&\leq \LR^\varepsilon(\mc N)+1.
	\end{aligned}
	\end{equation}

\end{proof}
\begin{remark}\label{oneshotexact}
Define the smooth robustness of quantum memory as
\begin{equation}
    \mc R^\varepsilon(\mc N)=\min_{\|\mc N'-\mc N\|_\diamond\leq \varepsilon}\mc R(\mc N')
\end{equation}
Since dimensions are integers, the single-shot memory dilution rate can be exactly characterised as
\begin{equation}
    R_c^{1,\varepsilon}(\mc N)=\log_2\left(1+\lceil \mc R^\varepsilon(\mc N)\rceil\right).
\end{equation}
\end{remark}

The above can be compared with the one-shot characterisation of dilution in resource theories of states~\cite{Brandao2011oneshot,zhao2018oneshot,liu2019resource}, which yields related results but is not applicable to the study of manipulation of quantum channels.

\section{Simulating quantum memories with classical resources}
The  task is to simulate a target quantum memory $\mc N$ with free EB memories and free super-operations. For an unknown input state $\rho$, the output state is $\mc N(\rho)$. Suppose after another operation $\mc U$ is applied to the output state, we read out the system by the measuring the average value $\braket{O}_{\mc U\circ\mc N(\rho)}=\tr[\mc U\circ\mc N(\rho)O]$. 
The simulation scheme works as follows. 
\begin{enumerate}
    \item As we require the simulation scheme to be independent of $\rho$, $\mc U$, and $O$, it is equivalent to have $\mc N$ as a linear expansion of EB channels $\mc M_i$, i.e., 
    \begin{equation}\label{Eq:decomposition}
        \mc N=\sum_i c_i\mc M_i
    \end{equation}
    with coefficients $c_i\in\mathbb{R}$ possibly negative. This linear decomposition always exists as the set of EB channels forms a complete basis for the space of quantum channels. 
    \item To obtain the averaged value $\braket{O}_{\mc U\circ\mc N(\rho)}$, we first re-express it as $\braket{O}_{\mc U\circ\mc N(\rho)}=\|c\|_1\braket{\overline{O}}_{\mc U\circ\mc N(\rho)}$ with a normalised observable
    \begin{equation}\label{Eq:simulationdecom}
        \braket{\overline{O}}_{\mc U\circ\mc N(\rho)} = \sum_i p_i \textrm{sign}(c_i) \braket{O}_{\mc U\circ\mc M_i(\rho)},
    \end{equation}
    an overhead $\|c\|_1 = \sum_i|c_i|$, and a normalised probability distribution $\{p_i=|c_i|/\|c\|_1\}$.
    
    We can obtain $\braket{\overline{O}}_{\mc U\circ\mc N(\rho)}$ by averaging $\textrm{sign}(c_i)\braket{O}_{\mc U\circ\mc M_i(\rho)}$  with probability $p_i$, which is described as follows:
    \begin{enumerate}
        \item We randomly generate $i$ according to the probability distribution $\{p_i\}$.
        \item As each $\mc M_i$ is EB, we have $\mc M_i(\rho) = \sum_j\tr[\rho M_i^j]\sigma^j_i$ with POVM $\{M_i^j\ge 0\}$ satisfying $\sum_jM_i^j=I$.
        \item To get $\braket{O}_{\mc U\circ\mc M_i(\rho)}$, we first measure $\rho$ with the POVM $\{M_i^j\ge 0\}$, obtaining outcome $j$ with probability $\tr[\rho M_i^j]$. Then we prepare $\sigma^j_i$, apply $\mc U$, and measure the observable $O$ to have $\tr[\mc U(\sigma^j_i) O]$. 
        The value $\braket{O}_{\mc U\circ\mc M_i(\rho)}$ can be  evaluated by averaging the measurement results $\tr[\mc U(\sigma^j_i) O]$ over all outcomes $j$. 
        \item Finally, by multiplying $\textrm{sign}(c_i)$ to $\braket{O}_{\mc U\circ\mc M_i(\rho)}$ and averaging over all $\mc M_i$ with probability $p_i$, we recover the normalised average expectation value $\braket{\overline{O}}_{\mc U\circ\mc N(\rho)}$.
    \end{enumerate}
    \item    The target averaged value   $\braket{O}_{\mc U\circ\mc N(\rho)}$ is obtained by multiplying the constant overhead $\|c\|_1$ to $\braket{\overline{O}}_{\mc U\circ\mc N(\rho)}$. 
\end{enumerate}

Suppose we aim to estimate $\braket{O}_{\mc N(\rho)}$ to an additive error $\varepsilon$ with probability $\delta$, we need the number of samples to be
\begin{equation}
	T \propto \frac{\|c\|_1^2}{\varepsilon^2}\log(\delta^{-1}),
\end{equation}
according to the Hoeffding inequality.
Meanwhile, given the channel $\mc N$ itself, the number of samples needed  is $T_0 \propto 1/{\varepsilon^2}\log(\delta^{-1})$. Thus the simulation cost can be quantified by the overhead 
\begin{equation}
	T/T_0 = \|c\|_1^2. 
\end{equation}
We can further minimise the simulation cost over all possible decomposition strategies of Eq.~\eqref{Eq:decomposition}. Denote the positive and negative coefficients of $c_i$ by $c_i^+$ and $c_i^-$, respectively. Then we have 
\begin{equation}
    \mc N = \sum_{i:c_i\ge 0} |c_i^+| \mc M_i - \sum_{i:c_i< 0} |c_i^-| \mc M_i, \, \mc M_i\in \mc F,
\end{equation}
with $\|c\| = \sum_{i:c_i\ge 0}|c_i^+|+\sum_{i:c_i< 0}|c_i^-|$. As the channel is trace preserving, we also have $\sum_{i:c_i\ge 0}|c_i^+|-\sum_{i:c_i< 0}|c_i^-|=1$. Denote $s = \sum_{i:c_i< 0}|c_i^-|$ hence with $\|c\| = 2s+1$, $\mc M = \sum_{i:c_i< 0}|c_i^-| \mc M_i$, and $\mc M' = \sum_{i:c_i> 0}|c_i^+| \mc M_i$, we have
\begin{equation}
    \mc N = (s+1) \mc M - s \mc M', \, \mc M,\mc M'\in \mc F.
\end{equation}
Optimising over all possible decomposition is equivalent to optimising over all $\mc M,\mc M'\in \mc F$, which coincides with the definition of the robustness of memories. Therefore we have
\begin{equation}
	\min T/T_0 = (2\mc R(\mc N)+1)^2.
\end{equation}

Therefore, the robustness of quantum memory quantifies the cost of simulating the memory with free EB memories.


%
%
%

The result of this section can be regarded as an extension of the framework of negativity-based simulators of~\cite{PhysRevLett.115.070501,PhysRevLett.118.090501}, which were recently adapted to the study of general resource theories of states~\cite{seddon_2020}, and applied in other settings to investigate the properties of specific channel resources~\cite{seddon2019quantifying,wang2019quantifying}.

\section{Quantum game for testing the power of a memory}

In this section, we discuss how to use quantum games to test the power of a memory and how the optimal strategy is related to the robustness of memories. 
We consider quantum games similar to Ref.~\cite{memoryResource18prx} by firstly inputting  a
general set of input states $\{\sigma_i\}$ to the channel. In Ref.~\cite{memoryResource18prx}, the output state is measured together with an ancillary set of states, so that the test can be independent of whether the measurement is faithfully implemented. In this work, instead of requiring such a measurement-device-independent feature, we directly measure the output states with a general set of observables $\{O_j\}$ by assuming that the measurement is trusted. We can define a general pay-off function as
\begin{equation}
	\mc P(\mc N, \mc G) = \sum_{i,j} \alpha_{i,j} \tr[\mc N(\sigma_i)O_j],
\end{equation}
with real coefficients $\alpha_{i,j}$, where we use $\mc G = (\{\sigma_i\}, \{O_j\}, \{\alpha_{i,j}\})$ to denote a particular game. 
When the coefficients $\alpha_{i,j}$ are selected randomly,  the pay-off function can be arbitrary, and in particular does not have to be non-negative. Instead, we first constrain ourselves to the case where the coefficients are selected such that the pay-off function is non-negative for any channel, $\mc P(\mc N, \mc G) \geq 0$.
Then the maximal pay-off of the game is 
\begin{equation}\label{Eq:gameresult1}
	 \max_{\mc G\in\mc S_G} \mc P(\mc N, \mc G)= \mc R_G(\mc N) + 1,
\end{equation}
where the maximisation is over all games $\mc G\in \mc S_G$ with $\mc S_G = \{\mc G: \mc P(\mc N, \mc G) \geq 0,\, \mc P(\mc M, \mc G)\le 1, \, \forall\mc N\in\mathrm{CPTP}, \mc M\in \mc F\}$.

To prove it, we first briefly review the duality in conic optimisation. We follow the description of Ref.~\cite{takagi2019general} and refer to the references therein for more details.
Given real complete normed vector spaces $\mc W$ and $\mc W'$, a conic optimisation problem is defined as
\begin{equation}\label{Eq:primeform}
    p =\inf \{\langle A, x\rangle | \Lambda(x)=y, x \in \mathcal{K}\},
\end{equation}
where $A \in \mathcal{W}^{*}, y \in \mathcal{W}^{\prime}$, $\Lambda : \mc{W} \rightarrow \mc{W}^{\prime}$ is a linear function, and $\mathcal{K} \subseteq \mathcal{W}$ is a closed and convex cone.
The dual form of the optimisation is given by
\begin{equation}
d=\sup \left\{\langle Z, y\rangle | A-\Lambda^{*}(Z) \in \mathcal{K}^{*}\right\},
\end{equation}
where
\begin{equation}
\mathcal{K}^{*}=\left\{Y \in \mathcal{W}^{*} |\langle Y, k\rangle \geq 0,\forall k \in \mathcal{K}\right\}.
\end{equation}
The primal and dual problems are equivalent if Slater's condition is satisfied: that is there exists a feasible solution $x$ such that $x$ is in the (relative) interior of $\mc K$.

Now we prove Eq.~\eqref{Eq:gameresult1}.

\begin{proof}
We first write $\mc R_G(\mc N)$ in terms of its Choi state as
\begin{equation}
	\mc R_G(\mc N) = \min\left\{s\ge 0: \exists\mc M\in\mathrm{CPTP}, \mc M'\in \mc F, \, \textrm{s.t.} \, \frac{1}{s+1}\Phi^+_{\mc N}+\frac{s}{s+1}\Phi^+_{\mc M}=\Phi^+_{\mc M'}\right\}.
\end{equation}
This can be equivalently recast as
\begin{equation}
\begin{aligned}
	\mc R_G(\mc N)+1 =& \min\tr(x_1)\\
	\text{s.t. }& x_1-x_2 = \Phi^+_{\mc N},\\
	&x_1\in\mathrm{cone}(\mc F), x_2\in\mathrm{cone}(\mc V),
\end{aligned}
\end{equation}
where we use $\mc V$ ($\mc F$) to represent bipartite (separable) Choi states, while $\mathrm{cone}(\mc F)$ and $\mathrm{cone}(\mc V)$ represent their unnormalised versions. Now define  $\mc W=\mathrm{cone}(\mc V)\oplus \mathrm{cone}(\mc V)$, $\mathcal{W}^{\prime}=\mathcal{V}$, $\mathcal{K}=\mathrm{cone}(\mc F) \oplus \mathrm{cone}(\mc V)$, $\Lambda\left(x_{1} \oplus x_{2}\right)=x_{1}-x_{2}$, $A=I\oplus 0$, $y = \Phi^+_{\mc N}$, then $\mc R_G(\mc N)$ can be represented as the form of Eq.~\eqref{Eq:primeform}. Note that $\mc W^*$ is the set of Hermitian operators. The dual form of the optimisation gives
\begin{equation}
\begin{aligned}
	\mc R_G^d(\mc N)+1 = &\max \tr[W\Phi^+_{\mc N}]\\
	\text{s.t. }&\braket{I\oplus 0-\Lambda^*(W),k}\ge0,\forall k\in\mc K,\\
	&W\in\mc W^*,
\end{aligned}
\end{equation}
where 
\begin{equation}\label{Eq:primecondition}
    \braket{I\oplus 0-\Lambda^*(W),k}\ge0 ,\forall k\in\mc K\Leftrightarrow \tr[x_1]\ge\tr[Wx_1]-\tr[Wx_2],\forall x_1\in\mathrm{cone}(\mc F),\forall x_2\in\mathrm{cone}(\mc V).
\end{equation}
Note that the above condition is further equivalent to 
\begin{equation}\label{Eq:primeconditionequiv}
    \begin{aligned}
        \left(\tr[Wx_1] \le \tr[x_1], \forall x_1\in\mathrm{cone}(\mc F)\right)\wedge\left(\tr[Wx_2]\ge 0,\,  \forall x_2\in\mathrm{cone}(\mc V)\right).
    \end{aligned}
\end{equation}
This is because when we have $\tr[x_1]\ge\tr[Wx_1]-\tr[Wx_2],\,\forall x_1\in\mathrm{cone}(\mc F),\,\forall x_2\in\mathrm{cone}(\mc V)$, we can set $x_2=0$ to have $\tr[Wx_1] \le \tr[x_1], \,\forall x_1\in\mathrm{cone}(\mc F)$. Similarly, we set $x_1=0$ and get $\tr[Wx_2]\ge 0,\,  \forall x_2\in\mathrm{cone}(\mc V)$.
On the other hand, it is straightfward to verify that Eq.~\eqref{Eq:primeconditionequiv} implies Eq.~\eqref{Eq:primecondition}. Therefore, the dual form can be written as
\begin{equation}
\begin{aligned}
	\mc R_G^d(\mc N)+1 = &\max\tr\left[W\Phi^+_{\mc N}\right]\\
	\text{s.t. }&\tr[Wx_1] \le \tr[x_1], \forall x_1\in\mc F\\
	&\tr[Wx_2]\ge 0,\forall x_2\in\mc V\\
	&W\in\mc W^*,
\end{aligned}
\end{equation}
which can be expressed with quantum channels
\begin{equation}
\begin{aligned}
    \mc R_G^d(\mc N)+1= &\max\tr[\Phi^+_{\mc N} W]\\
\text{s.t. }&W^\dag = W,\\
            &\tr[\Phi^+_{\mc M} W]\ge 0,\\
            &\tr[\Phi^+_{\mc M'} W]\le 1,\\
            &\forall \mc M\in\mathrm{CPTP}, \mc M'\in \mc F.
\end{aligned}
\end{equation}
Slater's condition holds as we can choose, say, $W=I/2$. Thus we have  strong duality, $\mc R_G^d(\mc N)=\mc R_G(\mc N)$.

Now we show that the dual form is equivalent to the maximal pay-off of the quantum game. 
We first write the pay-off function in terms of the Choi state of the channel as
    \begin{equation}
	\mc P(\mc N, \mc G) = d\sum_{i,j} \alpha_{i,j} \tr[\Phi^+_{\mc N}(\sigma_i^T\otimes O_j)] = \tr[\Phi^+_{\mc N} W],
\end{equation}
where $\sigma_i^T$ is the transpose of $\sigma_i$ and $W$ is a Hermitian operator
\begin{equation}
    W = d\sum_{i,j} \alpha_{i,j} \sigma_i^T\otimes O_j.
\end{equation}
Denote $\mc P_{\max}(\mc N)=\max_{\mc G\in\mc S_G} \mc P(\mc N, \mc G)$ with $\mc S_G = \{\mc G: \mc P(\mc N, \mc G) \geq 0,\, \mc P(\mc M, \mc G)\le 1, \, \forall\mc N\in\mathrm{CPTP}, \mc M\in \mc F\}$, the optimisation over all games in $\mc S_G$ is equivalent to optimise over all Hermitian operators $W$ that satisfies $\tr[\Phi^+_{\mc M} W]\ge 0$ for all channels $\mc M$ and $\tr[\Phi^+_{\mc M'} W]\le 1$ for entanglement breaking channels. Therefore, we have 
\begin{equation}
\begin{aligned}
    \mc P_{\max}(\mc N)= \mc R_G^d(\mc N) + 1=\mc R_G(\mc N) + 1.
\end{aligned}
\end{equation}



\end{proof}

The robustness $\mc R(\mc N)$ can be understood very similarly in this context. The maximal pay-off of the game is 
\begin{equation}\label{Eq:testRQM2}
\max_{\mc G\in\mc S} \mc P(\mc N, \mc G)=	\mc R(\mc N) + 1,
\end{equation}
where the maximisation is over all games $\mc G\in \mc S$ with $\mc S = \{\mc G: \mc P(\mc M, \mc G)\in  [0,1], \, \forall \mc M\in \mc F\}$. 

\begin{proof}
We follow the proof for the generalised robustness. We first write the robustness measure in the standard form of the primal optimisation problem as
\begin{equation}
\begin{aligned}
	\mc R(\mc N)+1 =& \min\tr(x_1)\\
	\text{s.t. }& x_1-x_2 = \Phi^+_{\mc N},\\
	&x_1\in\mathrm{cone}(\mc F), x_2\in\mathrm{cone}(\mc F),
\end{aligned}
\end{equation}
where we define components in the standard form Eq.~\eqref{Eq:primeform} as $\mc W=\mathrm{cone}(\mc V)\oplus \mathrm{cone}(\mc V)$, $\mathcal{W}^{\prime}=\mathcal{V}$, $\mathcal{K}=\mathrm{cone}(\mc F) \oplus \mathrm{cone}(\mc F)$, $\Lambda\left(x_{1} \oplus x_{2}\right)=x_{1}-x_{2}$, $A=I\oplus 0$, $y = \Phi^+_{\mc N}$.

The dual form of the optimisation gives
\begin{equation}
\begin{aligned}
	\mc R^d(\mc N)+1 = &\max \tr[W\Phi^+_{\mc N}]\\
	\text{s.t. }&\braket{I\oplus 0-\Lambda^*(W),k}\ge0,\forall k\in\mc K,\\
	&W\in\mc W^*,
\end{aligned}
\end{equation}
where 
\begin{equation}\label{Eq:primecondition2}
    \braket{I\oplus 0-\Lambda^*(W),k}\ge0 ,\forall k\in\mc K\Leftrightarrow \tr[x_1]\ge\tr[Wx_1]-\tr[Wx_2],\forall x_1, x_2\in\mathrm{cone}(\mc F).
\end{equation}

Note that the above condition is further equivalent to 
\begin{equation}\label{Eq:primeconditionequiv2}
    \begin{aligned}
        \left(\tr[Wx_1] \le \tr[x_1], \forall x_1\in\mathrm{cone}(\mc F)\right)\wedge\left(\tr[Wx_2]\ge 0,\,  \forall x_2\in\mathrm{cone}(\mc F)\right).
    \end{aligned}
\end{equation}
We can express it with quantum channels as
\begin{equation}
\begin{aligned}
    \mc R^d(\mc N)+1= &\max\tr[\Phi^+_{\mc N} W]\\
    \text{s.t. }&W=W^\dag,\\
                &\tr[\Phi^+_{\mc M'} W]\in [0,1], \forall \mc M'\in \mc F,
\end{aligned}
\end{equation}
which is exactly the maximal pay-off function $\mc P_{\max}(\mc N)=\max_{\mc G\in \mc S} \mc P(\mc N, \mc G)$. Since Slater's condition also holds in this case as we can choose $W=I/2$, we conclude that strong duality holds, thus $ \mc P_{\max}(\mc N)= \mc R(\mc N) + 1.$


\end{proof}
By considering  games $\mc G\in \mc S'$ with $\mc S' = \{\mc G: \mc P(\mc M, \mc G)\ge 0, \, \forall \mc M\in \mc F\}$,
we thus get the result presented in the main text,
\begin{equation}
\max_{\mc G\in\mc S'} \frac{P(\mc N, \mc G)}{\max_{\mc M \in \mc F}\mc P (\mc M, \mc G)} = \mc R(\mc N) + 1.
\end{equation}
 
The characterisation of $\mc R_G$ in terms of performance in nonlocal games can be compared with~\cite{takagi2018operational,takagi2019general}, where this quantity was related to the advantage in discrimination tasks in general resource theories of states and channels. The setting considered in that work is different, however, most importantly since the measurements there are performed on the joint Choi state of channels rather than their output states, making them significantly more difficult to perform in practice.

\section{Numerical and analytical calculation}\label{App:numerical}

\subsection{Relaxing to PPT-inducing channels}
The robustness and generalised robustness measures of memories are both convex optimisation problems due to the convexity property. They can be numerically calculated by focusing on the Choi states as
 
 \begin{equation}
\begin{split}
  \mc R(\mc N)=&\min s \\
  \text{ s.t. } &s\ge 0, \\
                &\mc M,\mc M'\in \mc F,\\ &\frac{1}{s+1}\Phi^+_{\mc N}+\frac{s}{s+1}\Phi^+_{\mc M}=\Phi^+_{\mc M'}.\\
\end{split}
\end{equation}

and
 \begin{equation}
\begin{split}
  \mc R_G(\mc N)=&\min s \\
  \text{ s.t. } &s\ge 0, \\ 
                &\mc M'\in \mc F,\\
                &\mc M \in \mathrm{CPTP},\\
                &\frac{1}{s+1}\Phi^+_{\mc N}+\frac{s}{s+1}\Phi^+_{\mc M}=\Phi^+_{\mc M'} .\\
\end{split}
\end{equation}
 
Instead of directly solving this optimisation problem, we reduce the restrictions that
$\Phi^+_{\mc M},\Phi^+_{\mc M'}$ are separable Choi states to ones that they are positive partial transpose (PPT) Choi states. 
We denote the robustness measures against PPT-inducing channels as $\mc R^*(\mc N)$ and $\mc R^*_G(\mc N)$, defined as
 \begin{equation}
\begin{split}
  \mc R^*(\mc N)=&\min s \\
  \text{ s.t. } &s\ge 0, \\
                &\Phi^+_{\mc M}, \Phi^+_{\mc M'} \in \mathrm{PPT}, \\ &\frac{1}{s+1}\Phi^+_{\mc N}+\frac{s}{s+1}\Phi^+_{\mc M}=\Phi^+_{\mc M'},
\end{split}
\end{equation}
and
 \begin{equation}
\begin{split}
  \mc R^*_G(\mc N)=&\min s \\
  \text{ s.t. } &s\ge 0, \\
                &\Phi^+_{\mc M'} \in \mathrm{PPT}, \\
                &\mc M \in \mathrm{CPTP}, \\
                &\frac{1}{s+1}\Phi^+_{\mc N}+\frac{s}{s+1}\Phi^+_{\mc M}=\Phi^+_{\mc M'}.
\end{split}
\end{equation}
As the set of PPT Choi states is larger than the set of separable Choi states, the optimisation always give lower bounds
\begin{equation}
    \mc R^*(\mc N)\le \mc R(\mc N),\, \mc R^*_G(\mc N)\le \mc R_G(\mc N).
\end{equation}
The bound is tight for channels with qubit inputs and outputs or ones with qubit inputs and qutrit outputs or vice versa.
In general cases, the optimisation with respect to PPT Choi states can be efficiently solved as semidefinite programs. In the case of qubit outputs, we will in fact establish an analytical form of the measures.
 \begin{equation}
\begin{split}
  \mc R^*(\rho_{AB}=\Phi^+_{\mc N})=&\min \tr[M_{AB}] \\
  \text{ s.t. } &s\ge0,\\
                &\tr_B[M_{AB}]=sI_A,\\ 
                &M_{AB},M'_{AB}\ge0,\\ 
                &(M_{AB})^{T_A}, ({M'_{AB}})^{T_A}\ge0,\\ 
                &\rho_{AB}+M_{AB}=M'_{AB},
\end{split}
\end{equation}
and
 \begin{equation}
\begin{split}
  \mc R_G^*(\rho_{AB}=\Phi^+_{\mc N})=&\min \tr[M_{AB}] \\
  \text{ s.t. } &s\ge0,\\
                &\tr_B[M_{AB}]=sI_A,\\ 
                &M_{AB},M'_{AB}\ge0,\\ 
                &({M'_{AB}})^{T_A}\ge0,\\ 
                &\rho_{AB}+M_{AB}=M'_{AB}.
\end{split}
\end{equation}
In Fig.~\ref{Fig:simulation}(a) we show several the robustness of examples channels.

Meanwhile the expression for the robustness of quantum memories $\mc N$ is similar to the robustness of entanglement of Choi matrix $\Phi^+_{\mc N}$, which is defined as

 \begin{equation}
\begin{split}
  RE(\Phi^+_{\mc N})=&\min s \\
  \text{ s.t. } &s\ge 0, \\
                &\sigma,\sigma'\in \mathrm{SEP},\\ &\frac{1}{s+1}\Phi^+_{\mc N}+\frac{s}{s+1}\sigma=\sigma',
\end{split}
\end{equation}
and
 \begin{equation}
\begin{split}
  RE_G(\Phi^+_{\mc N})=&\min s \\
  \text{ s.t. } &s\ge 0, \\
                &\sigma\in\mc D,\\
                &\sigma'\in \mathrm{SEP},\\ &\frac{1}{s+1}\Phi^+_{\mc N}+\frac{s}{s+1}\sigma=\sigma'.\\
\end{split}
\end{equation}
Note that $\Phi^+_{\mc M}, \Phi^+_{\mc M'}$ belong to the subset of  separable states whose partial trace is identity  due to the property of Choi states. Thus 
the robustness of entanglement of the Choi state lower bounds the robustness of quantum memory, $RE(\Phi^+_{\mc N})\le \mc R(\mc N)$ and $RE_G(\Phi^+_{\mc N})\le \mc R_G(\mc N)$. Similarly, we can also relax the separable state set with PPT to efficiently calculate $ RE^*(\Phi^+_{\mc N})$ and $ RE^*_G(\Phi^+_{\mc N})$, so that we have  $RE^*(\Phi^+_{\mc N})\le RE(\Phi^+_{\mc N})$ and $RE^*_G(\Phi^+_{\mc N})\le RE_G(\Phi^+_{\mc N})$.
By randomly choosing quantum channels, one can find that the quantities  $ RE^*(\Phi^+_{\mc N})$ and $\mc R^*(\mc N)$ are different, as shown in Fig.~\ref{Fig:simulation}(b) and (d). 

\begin{figure}[t]
\centering
\minipage{.24\textwidth}
  \includegraphics[width=1\linewidth]{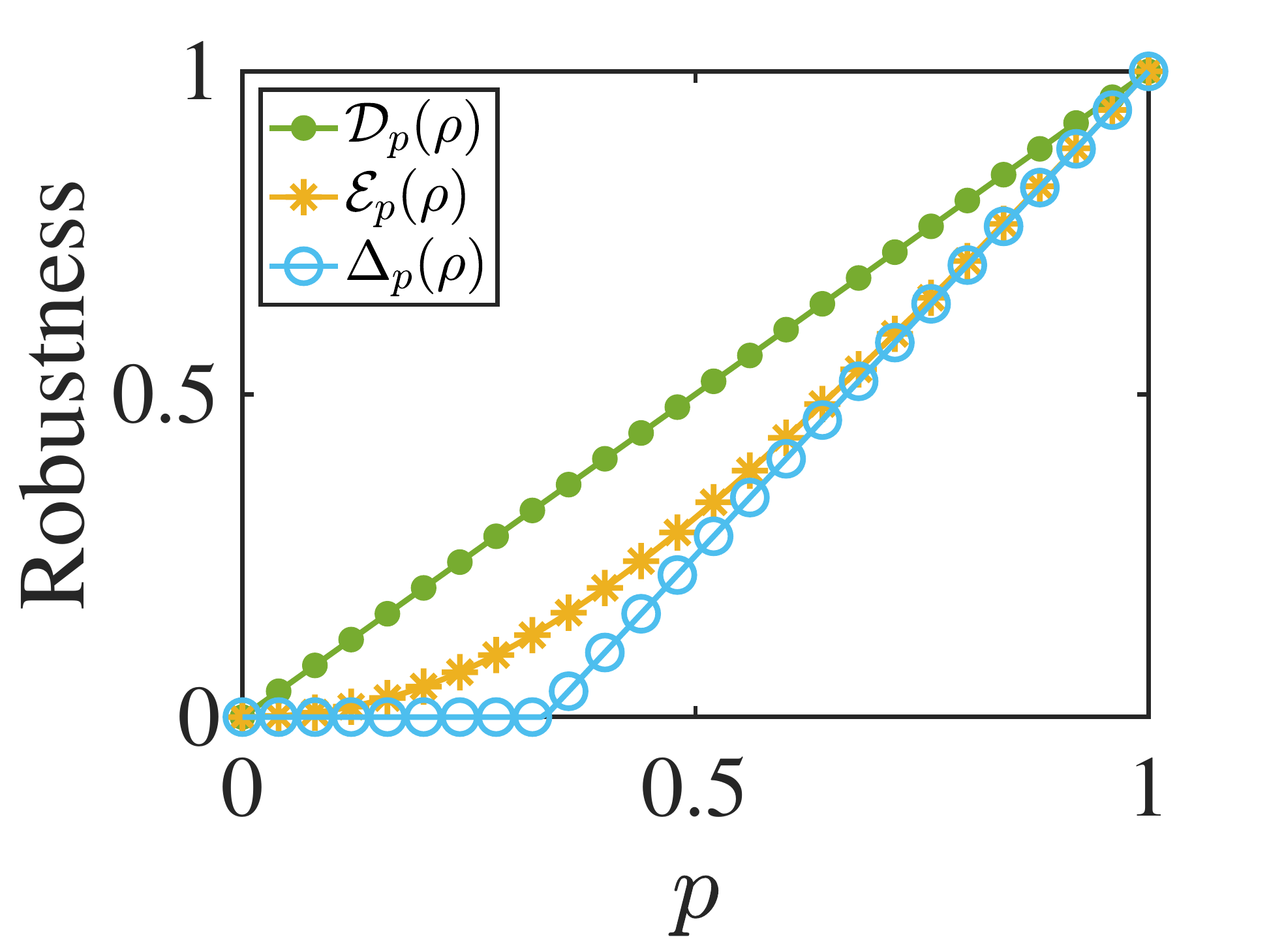}
  (a)\label{fig:awesome_image1}
\endminipage\hfill
\minipage{.24\textwidth}
  \includegraphics[width=1\linewidth]{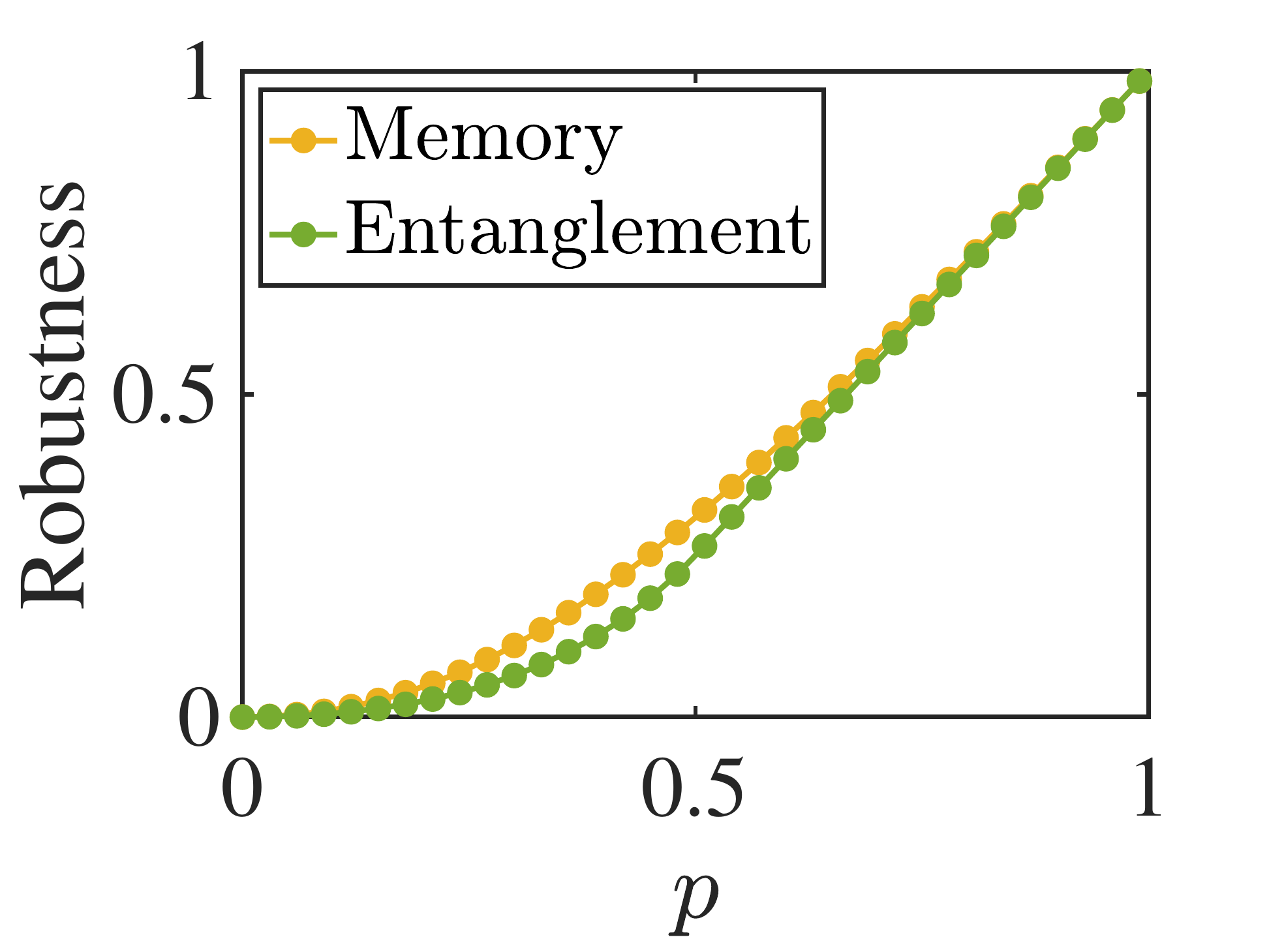}
  (b)\label{fig:awesome_image3}
\endminipage\hfill  
%
\minipage{.24\textwidth}
  \includegraphics[width=1\linewidth]{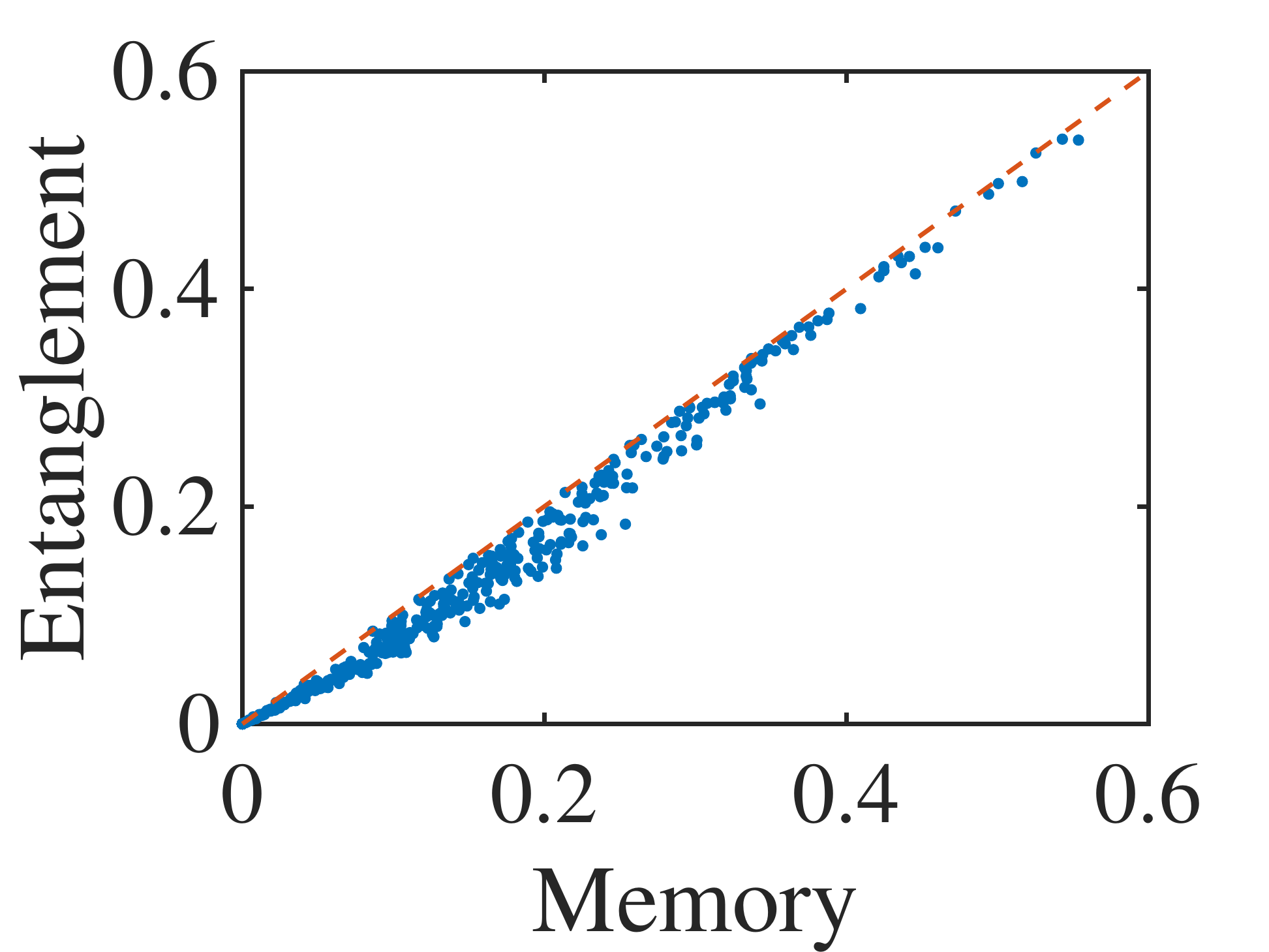}
  (c)\label{fig:awesome_image4}
\endminipage\hfill  
\minipage{.24\textwidth}
  \includegraphics[width=1\linewidth]{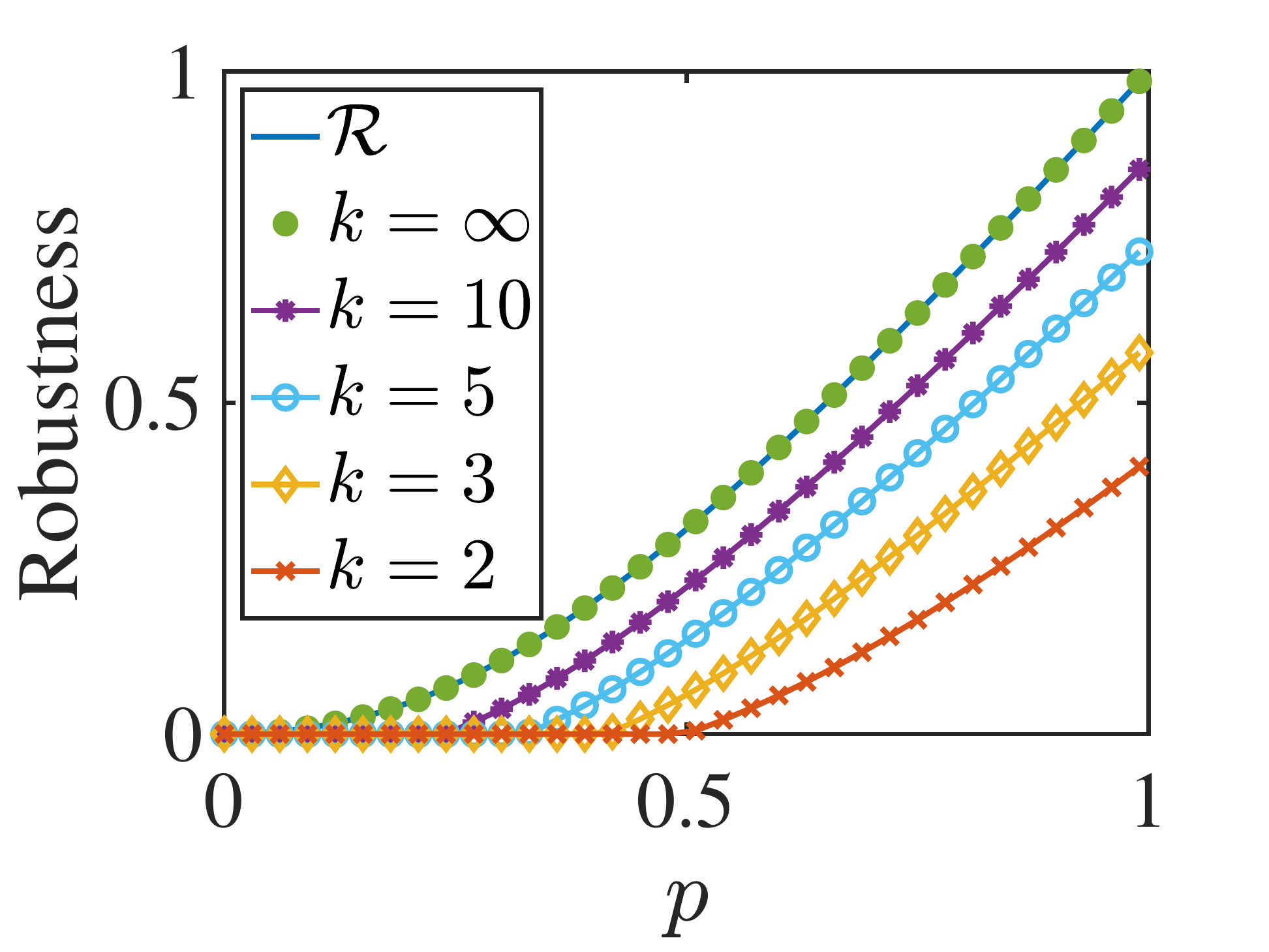}
  (d)\label{fig:awesome_image2}
\endminipage\hfill  
  \caption{
  Numerical calculation of robustness measures of memories and entanglement.
  (a) Robustness of memories with qubit inputs and computational basis $\{\ket{0},\ket{1}\}$ for depolarising channels $\Delta_p(\rho)=p\rho+(1-p)I/2$, stochastic damping channels $\mathcal D_p(\rho) = p\rho+(1-p)|{0}\rangle\langle{0}|$, and erasure channels $\mathcal E_p(\rho) = p\rho+(1-p)|{2}\rangle\langle{2}|$ with $\ket{2}$ orthogonal to $\{\ket{0},\ket{1}\}$. 
%
%
  (b) Comparison between the robustness of memory and the robustness of entanglement of stochastic damping channels. Here the entanglement of a channel refers to the entanglement of the corresponding Choi state. 
  (c) Comparison between the robustness of memories (horizontal axis) and the robustness of entanglement (vertical axis) of random channels with qubit input and output.
  (d) Estimation of the RQM via different moments of the Choi state given in Eq.~\eqref{Eq:momentsEstApp} for stochastic damping channels.  In particular, the estimation is tight with the infinite moment which corresponds to the maximal eigenvalue of the Choi state.
  }\label{Fig:simulation}
\end{figure}

Note that for the numerical examples in the main text, we focus on quantum channels with input dimension $d_{\text{in}}$ and output dimension $d_{\text{out}}$ such that $d_{\text{in}}\times d_{\text{out}}\leq 6$, in which case the SDP relaxation provided here is actually tight, due to the well-known correspondence between SEP and PPT in low dimensions~\cite{horodecki_1996-1}.



\subsection{Examples of quantum game}

The pay-off of a quantum game $\mc G = (\{\sigma_i\}, \{O_j\}, \{\alpha_{i,j}\})$
can be expressed as
\begin{equation}
\begin{aligned}
	\mc P(\mc N, \mc G) &= \sum_{i,j} \alpha_{i,j} \tr[\mc N(\sigma_i)O_j]\\
	 &=\sum_{i,j} \alpha_{i,j} d\tr[\Phi_{\mc N}^+(\sigma_i^T\otimes O_j)]\\
	 &=\tr[\Phi_{\mc N}^+W],
\end{aligned}
\end{equation}
where $\Phi_{\mc N}^+=\mc I\otimes \mc N(\Phi^+)$ is the Choi state of $\mc N$,  $\Phi^+=1/d\sum_{ij}\ket{ii}\bra{jj}$ is the maximally entangled state, $d$ is the dimension of the input system, and the game operator is
\begin{equation}
    W = d\sum_{i,j} \alpha_{i,j} \sigma_i^T\otimes O_j.
\end{equation}
For the robustness of quantum memories, we consider games that satisfy 
\begin{equation}
    \mc P(\mc M, \mc G)\in [0,1], \forall \mc M\in \mc F.
\end{equation}

\subsubsection{Depolarising channel $\Delta_p(\rho)=p\rho+(1-p)I/2$}
For the depolarising channel $\Delta_p(\rho)=p\rho+(1-p)I/2$, the corresponding Choi state is the Werner state,
\begin{equation}
   \Phi_{\Delta_p}^+ = p\Phi^+ + (1-p)I/4;
\end{equation}
The robustness of $\Delta_p(\rho)$ can be obtained by applying the entanglement witness considered in our previous work~\cite{PhysRevLett.112.140506}. Here, we construct the game $\mc G = (\{\sigma_i\}, \{O_j\}, \{\alpha_{i,j}\})$ such that
\begin{equation}
    W = 2\Phi^+.
\end{equation}
Because $2\tr[\Phi^+\sigma]\le1$ for all bipartite qubit separable states $\sigma$, we have $P(\mc M, \mc G)\in  [0,1]$, $\forall \mc M\in \mc F$. We can also show that $W$ is the optimal witness of $\mc R_G(\Delta_p)$ by proving
\begin{equation}
    \mc R_G ( \Delta_p) \le \tr[W\Phi_{\Delta_p}^+]-1 = \frac{3p-1}{2}.
\end{equation}
This is equivalent to finding a separable Choi state $\Phi_{\mc M}^+$ that satisfies
\begin{equation}
    \Phi_{\Delta_p}^+ \le  \frac{3p+1}{2}\Phi_{\mc M}^+.
\end{equation}
This is satisfied by choosing $\Phi_{\mc M}^+=\frac{1}{3}\Phi^+ + \frac{1}{6}I$. 

Now we show how to decompose the witness into input states and measurements. Suppose we choose $\{\sigma_i\}, \{O_j\}$ as
    \begin{equation}\label{Eq:innputsapp}
    \begin{aligned}
    \sigma_0 = \frac{I}{2}, \,   \sigma_1 = \frac{I+\sigma_x}{2}, \, \sigma_2 = \frac{I-\sigma_y}{2},\\
    \sigma_3 = \frac{I+\sigma_z}{2}, \,
    \sigma_4 = \frac{I+(\sigma_x+\sigma_y+\sigma_z)/\sqrt{3}}{2},\\
    \end{aligned}
\end{equation}
with $O_i = \sigma_i^T$, the corresponding coefficients are
\begin{equation}
\label{Eq:coedeplor}
\alpha_{i,j} = \left[
  \begin{array}{ccccc}
    2\sqrt3 & 0 & 0 & 0 & -\sqrt3 \\
    0 & 1 & 0 & 0 & 0\\
    0 & 0 & -1 & 0 & 0 \\
    0 & 0 & 0 & 1  & 0\\
    -\sqrt3 & 0 & 0 & 0  & 0\\
  \end{array}
\right].
\end{equation}
For the qubit depolarising channel $\Delta_p(\rho)=p\rho+(1-p)I/2$, 
the pay-off of this game is
\begin{equation}
    \mc P( \Delta_p, \mc G) = \frac{3p}{2}+\frac{1}{2}.
\end{equation}
This provides a lower bound
\begin{equation}
\begin{aligned}
    \mc R(\Delta_p)&\ge\max \{\mc P(\Delta_p, \mc G)-1, 0\},\\
    &=\left\{
    \begin{array}{cc}
        (3p-1)/2 & p\in [1/3, 1] \\
        0 & p\in [0, 1/3)
    \end{array}
    \right .
\end{aligned}
\end{equation}
The equal sign is always achieved, as verified from the numerical calculation in Fig.~\ref{Fig:simulation}(a). 


\subsubsection{Erasure channels $\mathcal E_p(\rho) = p\rho+(1-p)|{2}\rangle\langle{2}|$}
Considering the erasure channels $\mathcal E_p(\rho) = p\rho+(1-p)|{2}\rangle\langle{2}|$ with $\ket{2}$ orthogonal to $\{\ket{0},\ket{1}\}$, the Choi state is
\begin{equation}
    \Phi_{\mathcal E_p}^+ = p\Phi^+ + (1-p)\frac{I}{2}\otimes \ket{2}\bra{2}.
\end{equation}
Now we choose the game with a witness 
\begin{equation}
    W = 2\Phi^++I\otimes\ket{2}\bra{2},
\end{equation}
and we can prove that $P(\mc M, \mc G)\in  [0,1]$, $\forall \mc M\in \mc F$ by showing $\tr[W\sigma]\in[0,1]$ for all separable states $\sigma$. We sketch the proof here.
Consider a projective measurement $\{P_{01}=\ket{0}\bra{0}+\ket{1}\bra{1}, P_2=\ket{2}\bra{2}\}$ on the second system, then 
\begin{equation}
\begin{aligned}
\tr[W\sigma] &= \tr[(P_{01}+P_2)W(P_{01}+P_2)\sigma],\\
&=\tr[(P_{01}+P_2)(2\Phi^++I\otimes\ket{2}\bra{2})(P_{01}+P_2)\sigma],\\
&=\tr[P_{01}2\Phi^+P_{01}\sigma]+\tr[P_{2}I\otimes\ket{2}\bra{2}P_{2}\sigma],\\
&=\tr[2\Phi^+P_{01}\sigma P_{01}]+\tr[I\otimes\ket{2}\bra{2}P_{2}\sigma P_{2}],\\
&=p_{01}\tr[2\Phi^+\sigma_{01}]+p_2\tr[I\otimes\ket{2}\bra{2}\sigma_2],\\
&=p_{01}\tr[2\Phi^+\sigma_{01}]+p_2,\\
&\le 1.
\end{aligned}
\end{equation}
Here $p_{01}=\tr[P_{01}\sigma_{01}P_{01}]$, $p_{2}=\tr[P_{2}\sigma_{2}P_{2}]$, $\sigma_{01}=P_{01}\sigma P_{01}/p_{01}$, and $\sigma_{2}=P_{2}\sigma P_{2}/p_{2}$. The last line follows from $\tr[2\Phi^+\sigma_{01}]\le 1$ and $p_{01}+p_2=1$. 

As $W\ge0$, $\tr[W\sigma]\ge 0$. We can check that
\begin{equation}
    \tr[W\Phi_{\mathcal E_p}^+] = p+1,
\end{equation}
which is consistent with the numerical calculation, 
\begin{equation}
\begin{aligned}
    \mc R_G (\mathcal E_p)&=p.
\end{aligned}
\end{equation}
We can also prove that the estimate from the game witness is tight. That is, we need to find a separable Choi state $\Phi_{\mc M}^+$ that satisfies
\begin{equation}
    \Phi_{\mc E_p}^+ \le  (p+1)\Phi_{\mc M}^+.
\end{equation}
Here, we choose $\Phi_{\mc M}^+=q(\frac{1}{3}\Phi^++ \frac{1}{6}I_2\otimes I_2)+ (1-q)\frac{I}{2}\otimes \ket{2}\bra{2}$ with $I_2=\ket{0}\bra{0}+\ket{1}\bra{1}$. To satisfy the above inequality, we choose $q = \frac{2p}{p+1}$ when $p\ge1/3$ and $q = \frac{1-p}{p+1}$ when $p\le 1/3$.

The decomposition of the witness is similar to the one for the depolarising channel. The only difference is that we need to consider an additional measurement to take into account of the term $I\otimes \ket{2}\bra{2}$. That is, with the same inputs given in Eq.~\eqref{Eq:innputsapp}, we consider the game with input states
\begin{equation}\label{Eq:innputsera}
    \begin{aligned}
    \sigma_0 = \frac{I}{2}, \,   \sigma_1 = \frac{I+\sigma_x}{2}, \, \sigma_2 = \frac{I-\sigma_y}{2},\\
    \sigma_3 = \frac{I+\sigma_z}{2}, 
    \sigma_4 = \frac{I+(\sigma_x+\sigma_y+\sigma_z)/\sqrt{3}}{2},\\
    \end{aligned}
\end{equation}
measurements 
\begin{equation}
    \begin{aligned}
    O_0 = \frac{I_2}{2}, \,   O_1 = \frac{I_2+\sigma_x}{2}, \, O_2 = \frac{I_2+\sigma_y}{2},\\
    O_3 = \frac{I_2+\sigma_z}{2}, 
    O_4 = \frac{I_2+(\sigma_x-\sigma_y+\sigma_z)/\sqrt{3}}{2}, O_5 = \ket{2}\bra{2}\\
    \end{aligned}
\end{equation}
and coefficients
\begin{equation}\label{Eq:coeferas}
\alpha_{i,j} = \left[
  \begin{array}{cccccc}
    2\sqrt3 & 0 & 0 & 0 & -\sqrt3 & 1\\
    0 & 1 & 0 & 0 & 0 & 0\\
    0 & 0 & -1 & 0 & 0 & 0\\
    0 & 0 & 0 & 1  & 0& 0\\
    -\sqrt3 & 0 & 0 & 0  & 0& 0\\
  \end{array}
\right].
\end{equation}

\subsubsection{stochastic damping channels $\mathcal D_p(\rho) = p\rho+(1-p)|{0}\rangle\langle{0}|$}

Note that the robustness of quantum memories of the two above examples are equal to the robustness of entanglement of the corresponding Choi states. For the stochastic damping channel, we found them different as shown in numerical examples of Fig.~\ref{Fig:simulation}(c) of the main text.


The Choi state of the stochastic damping channels $\mathcal D_p(\rho) = p\rho+(1-p)\ketbra{0}{0}$ is 
\begin{equation}
    \Phi_{\mc D_p }^+ = p\Phi^++(1-p)\frac{I}{2}\otimes\ketbra{0}{0}.
\end{equation}
Suppose the spectral decomposition of this matrix is 
\begin{equation}
\Phi_{\mc D_p }^+=\lambda_0 \psi_0+ \lambda_1 \psi_1+ \lambda_2 \psi_2 + \lambda_3 \psi_3,
\end{equation}
where $\lambda_0 \ge \lambda_1\ge \lambda_2\geq \lambda_3$ and $\psi_i$ is the density matrix of eigenstate. We can compute that
\begin{equation}
    \lambda_0=\frac{1+p+\sqrt{1-2p+5p^2}}{4}
\end{equation}
as well as $\ket{\psi_0}=\alpha\ket{00}+\beta\ket{11}$ with real coefficients $\alpha$ and $\beta$ depending on $p$,
\begin{equation}
    \begin{aligned}
    \alpha&=\frac{1-p+\sqrt{1-2p+5p^2}}{ \sqrt{\left(1-p+\sqrt{1-2p+5p^2}\right)^2+4p^2}},\\
    \beta&=\frac{2p}{\sqrt{\left(1-p+\sqrt{1-2p+5p^2}\right)^2+4p^2}}.
    \end{aligned}
\end{equation}

Now construct the game with a witness $W$ defined by
\begin{equation}
W=2\psi_0.
\end{equation}
Note that the witness does not satisfy $\tr[W\sigma]\in[0,1]$ for all separable states $\sigma$ as $\max_{\sigma\in\mathrm{SEP}}\tr[W\sigma]=2\max\{|\alpha|^2,|\beta|^2\}$ can be larger than 1. However, Choi states of EB channels are only a subset of separable states, and we can indeed show  that $\tr[W\sigma]\in[0,1]$ for all separable Choi states $\sigma$. That is, for all separable states $\sigma_{AB}$ satisfying $\tr_B[\sigma_{AB}]=I_A/2$, we want to prove
\begin{equation}
    \tr[\psi_0\sigma_{AB}]\le 1/2.
\end{equation}

We prove it by considering a stronger scenario for general EB channels with input dimension $d$ and we also consider an optimisation over all possible pure states $\psi_0$,
\begin{equation}
    \max_{\psi_0}\tr[\psi_0\sigma_{AB}]\le 1/d.
\end{equation}
This is equivalent to the following Lemma.
\begin{lemma}\label{Lemmainfty}
For any EB channel with input dimension $d$, the maximal eigenvalue of its Choi state $\sigma_{AB}$ is upper bounded by $1/d$,
\begin{equation}
    \max\textrm{eig}(\sigma_{AB})\le 1/d.
\end{equation}
\end{lemma}

\begin{proof}

The maximal eigenvalue of $\sigma_{AB}$ can be obtained by considering its Schatten $\infty$-norm. 
For an EB channel, its Choi state can always be expressed as
\begin{equation}
	\sigma_{AB} = \sum_i p_i\ketbra{\psi_i}{\psi_i}\otimes\ketbra{\phi_i}{\phi_i},
\end{equation}
with $\sum_i p_i \psi_i = I/d$, $\sum_i p_i = 1$, and $p_i\ge 0$.
Therefore we have
\begin{equation}
    \begin{aligned}
    \max\textit{eig}(\sigma_{AB}) = & \lim_{n\rightarrow\infty}\left(\tr[\sigma_{AB}^n]\right)^{\frac{1}{n}},\\
    =&\lim_{n\rightarrow\infty}\left(\tr\left[\sum_i p_{i_1}p_{i_2}\dots p_{i_n}\psi_{i_1}\psi_{i_2}\dots\psi_{i_n}\otimes\phi_{i_1}\phi_{i_2}\dots\phi_{i_n}\right]\right)^{\frac{1}{n}},\\
    =&\lim_{n\rightarrow\infty}\left(\sum_i p_{i_1}p_{i_2}\dots p_{i_n}\tr\left[\psi_{i_1}\psi_{i_2}\dots\psi_{i_n}\right]\tr\left[\phi_{i_1}\phi_{i_2}\dots\phi_{i_n}\right]\right)^{\frac{1}{n}},\\
    \le&\lim_{n\rightarrow\infty}\left(\sum_i p_{i_1}p_{i_2}\dots p_{i_n}\tr\left[\psi_{i_1}\psi_{i_2}\dots\psi_{i_n}\right]\right)^{\frac{1}{n}},\\
    =&\lim_{n\rightarrow\infty}\left(\tr\left[\left(\sum_ip_i\psi_i\right)^n\right]\right)^{\frac{1}{n}},\\
    =&\lim_{n\rightarrow\infty}\left(\frac{1}{d^{n-1}}\right)^{\frac{1}{n}},\\
    =&\frac{1}{d}.
    \end{aligned}
\end{equation}
Here the inequality follows from $\tr\left[\phi_{i_1}\phi_{i_2}\dots\phi_{i_n}\right]\le 1$.

\end{proof}

The decomposition of $\psi_0$ can be expressed as

\begin{equation}
\begin{aligned}
\psi_0=&4\sqrt{2}\alpha\beta \left(\frac{I}{2}\right)^{\otimes 2}+ \alpha^2\left(\frac{I+\sigma_z}{2}\right)^{\otimes 2} +\beta^2\left(\frac{I-\sigma_z}{2}\right)^{\otimes 2}+ 2\alpha\beta \left[\left(\frac{I+\sigma_x}{2}\right)^{\otimes 2}-\left(\frac{I+\sigma_y}{2}\right)^{\otimes 2}\right]\\
&-2\sqrt{2}\alpha\beta\left[ \left(\frac{I+\frac{\sigma_x-\sigma_y}{\sqrt{2}}}{2}\right)\otimes \left(\frac{I}{2}\right)+\left(\frac{I}{2}\right)\otimes \left(\frac{I+\frac{\sigma_x-\sigma_y}{\sqrt{2}}}{2}\right)\right].
\end{aligned}
\end{equation}
Thus we can choose 
\begin{equation}\label{Eq:inputsampli}
    \begin{aligned}
    O_0 = \frac{I}{2}, \,   O_1 = \frac{I+\sigma_x}{2}, \, O_2 = \frac{I+\sigma_y}{2},\,
    O_3 = \frac{I+\sigma_z}{2}, \\
    O_4 = \frac{I-\sigma_z}{2}, \, O_5 = \frac{I+(\sigma_x-\sigma_y)/\sqrt{2}}{2},\\
    \end{aligned}
\end{equation}
and  $ \sigma_i=O_i^T$with coefficients
\begin{equation}\label{Eq:coefampl}
\alpha_{i,j} = \left[
  \begin{array}{cccccc}
    4\sqrt{2}\alpha\beta & 0 & 0 & 0 & 0 & -2\sqrt{2}\alpha\beta\\
    0 & 2\alpha\beta & 0 & 0 & 0 & 0\\
    0 & 0 & -2\alpha\beta & 0 & 0 & 0\\
    0 & 0 & 0 & \alpha^2  & 0& 0\\
    0 & 0 & 0 & 0  & \beta^2& 0\\
  -2\sqrt{2}\alpha\beta & 0 & 0 & 0  & 0& 0\\
  \end{array}
\right].
\end{equation}
And finally we obtain the lower bound
\begin{equation}
    \mc R(\mc D_p)\geq 2\lambda_0-1=\frac{\sqrt{1-2p+5p^2}+p-1}{2}.
\end{equation}

\subsection{Tight lower bound via moments of Choi states}
Consider the Choi state $\Phi_{\mc N}^+$ of a channel $\mc N$, we can either lower bound the robustness via a witness measurement or via purity measurement. We prove Lemma 1 in the main text.

\begin{lemma}\label{lemma4}
For any EB channel $\mc M$ with input dimension $d$ and $k=0,1,\dots,\infty$, we have
\begin{equation}\label{Eq:firsthalfApp}
    \tr\left[\left(\Phi_{\mc M}^+\right)^k\right]\le {\frac{1}{d^{k-1}}}.
\end{equation}
For any channel $\mc N$ with input dimension $d$, its RQM can be lower bounded by 
\begin{equation}\label{Eq:momentsEstApp}
    \mc R(\mc N)\ge\mc R_G(\mc N) \ge d^{\frac{k-1}{k}}\left(\tr\left[\left(\Phi_{\mc N}^+\right)^k\right]\right)^\frac{1}{k}-1.
\end{equation}
\end{lemma}
\begin{proof}
The first half of this Lemma, Eq.~\eqref{Eq:firsthalfApp}, can be proven by following the proof of Lemma~\ref{Lemmainfty}.

For the second half, we show that $\mc R_G(\mc N) \ge d^{\frac{k-1}{k}}\left(\tr\left[\left(\Phi_{\mc N}^+\right)^k\right]\right)^\frac{1}{k}-1$. Suppose the optimal decomposition is
$$\Phi^+_{\mc N}=(s+1)\Phi^+_{\mc M'}-s\Phi^+_{\mc M},$$
with $\mc M'\in\mathrm{EB}$ and $\mc M\in\mathrm{CPTP}$.
Then
\begin{equation}
	\begin{aligned}
		\tr\left[\left(\Phi_{\mc N}^+\right)^k\right] \le(s+1)^k\tr\left[\left(\Phi^+_{\mc M'}\right)^k\right]
		\le \frac{(s+1)^k}{d^{k-1}},
	\end{aligned}
\end{equation}
and we obtain a lower bound for the robustness
\begin{equation}
	\mc R_G(\mc N) \ge d^{\frac{k-1}{k}}\left(\tr\left[\left(\Phi_{\mc N}^+\right)^k\right]\right)^\frac{1}{k}-1.
\end{equation}
\end{proof}
The comparison of the lower bounds with different moments of Choi states is shown in Fig.~\ref{Fig:simulation}(d).

We will show that the bound obtained by taking $k \to \infty$ is in fact tight for low-dimensional channels, and furthermore exactly characterises the PPT robustness $\mc R^*$ for all channels with qubit output.

\begin{theorem}\label{thm:rob_lower_bound_equal}
Consider a channel ${\mc N}: A \to B$ such that $d_A \leq 3$ and $d_B = 2$. Then
\begin{equation}\begin{aligned}
  {\mc R}({\mc N}) = {\mc R}_G({\mc N}) = \max\{ 0, \; d_A  \max \operatorname{eig} \Phi^+_{\mc N} - 1 \}.
\end{aligned}\end{equation}
More generally, for any channel such that $d_B = 2$, the PPT robustness satisfies
\begin{equation}\begin{aligned}
  {\mc R}^*({\mc N}) = {\mc R}_G^*({\mc N}) = \max\{ 0, \; d_A  \max \operatorname{eig} \Phi^+_{\mc N} - 1 \}.
\end{aligned}\end{equation}
\end{theorem}
\begin{proof}
We will employ the reduction criterion for separability~\cite{horodecki_1999,cerf_1999-1}, which states that a bipartite state $\rho_{AB}$ with $d_A \leq 3$ and $d_B = 2$ is separable if and only if
\begin{equation}\begin{aligned}
  \rho_{AB} \leq \tr_B(\rho_{AB}) \otimes I_B.
\end{aligned}\end{equation}
Notice that when $\rho_{AB} = \Phi^+_{\mc N}$ is the Choi state of a channel ${\mc N} : A \to B$, this reduces to
\begin{equation}\begin{aligned}
  \Phi^+_{\mc N} \leq \frac{1}{d_A} I_{AB},
\end{aligned}\end{equation}
that is, $\max \operatorname{eig} \Phi^+_{\mc N} \leq \frac{1}{d_A}$. The SDPs for the robustness measures then become
\begin{equation}\begin{aligned}
  {\mc R}_G({\mc N}) + 1 &= \min \left\{ \lambda \;\left|\; \Phi^+_{\mc N} \leq \lambda \Phi^+_{\mc M},\; \Phi^+_{\mc M} \leq \frac{1}{d_A} I,\; \Phi^+_{\mc M} \geq 0,\; \tr_B \Phi^+_{\mc M} = \frac{1}{d_A} I \right.\right\},\\
  {\mc R}({\mc N}) + 1 &= \min \left\{ \lambda \;\left|\; \Phi^+_{\mc N}  = \lambda \Phi^+_{{\mc M}_+} - (\lambda - 1) \Phi^+_{{\mc M}_-},\; \Phi^+_{{\mc M}_\pm} \leq \frac{1}{d_A} I,\; \Phi^+_{{\mc M}_\pm} \geq 0,\; \tr_B \Phi^+_{{\mc M}_\pm} = \frac{1}{d_A} I \right.\right\}.
\end{aligned}\end{equation}
The inequality ${\mc R}({\mc N}) \geq {\mc R}_G({\mc N}) \geq d_A \max \operatorname{eig} \Phi^+_{\mc N}$ is obvious as $\Phi^+_{\mc N} \leq \lambda \Phi^+_{\mc M} \leq \frac{\lambda}{d_A} I$ for any feasible $\lambda$ and $\Phi^+_{\mc M}$. On the other hand, let $\mu$ denote the largest eigenvalue of $\Phi^+_{\mc N}$. If $\mu \leq \frac{1}{d_A}$, then $\Phi^+_{\mc N}$ itself is a feasible solution and the equality ${\mc R}({\mc N}) = {\mc R}_G({\mc N}) = 0$ is trivial, so assume that $\mu \in (\frac{1}{d_A}, 1]$. Denoting by $\Delta$ the completely depolarising channel with $\Phi^+_\Delta = I / d_A d_B$, we define the channels ${\mc M}_\pm$ as the convex combinations
 \begin{equation}\begin{aligned}
 {\mc M}_- &:= \frac{1}{\mu d_A d_B - 1} \left( \mu d_A d_B \Delta - {\mc N} \right),\\
{\mc M}_+ &:= \frac{1}{\mu d_A}\left[ {\mc N} + (\mu d_A - 1) {\mc M}_- \right].
\end{aligned}\end{equation}
Noting that $\mu d_A \geq 1$ and thus $\mu d_A d_B - 1 \geq 2 \mu d_A - 1 \geq \mu d_A$, it follows that $\Phi^+_{{\mc M}_\pm} \leq \frac{1}{d_A} I$, which means that the above constitutes a valid feasible solution for the robustness ${\mc R}({\mc N})$ with $\lambda = \mu d_A$. We conclude that ${\mc R}_G({\mc N}) \leq {\mc R}({\mc N}) \leq \mu d_A = d_A \max \operatorname{eig} \Phi^+_{\mc N}$, and so the quantities must all be equal.

The second part of the Proposition follows since the action of the positive map $\rho \to (\tr \rho) I - \rho$, on which the reduction criterion is based, is unitarily equivalent to the transpose map when acting on a 2-dimensional system~\cite{horodecki_1999,cerf_1999-1}. This means that the reduction criterion and the PPT criterion are equivalent when $d_B = 2$.
\end{proof}

Altogether, the results establish Theorem 4 in the main text.

As an immediate consequence of the results, we notice that the multiplicativity of the lower bound coupled with the sub-multiplicativity of $R_G$ (see Sec.~\ref{sec:properties}) yields exact multiplicativity for $R_G$ or, equivalently, additivity for $D_{\max}$ in the case of low-dimension channels.

\begin{corollary}
For any channel $\mc N : A \to B$ with $d_B = 2$, we have
\begin{equation}
    \mc R^*_G(\mc N^{\otimes n}) + 1 = \left(\mc R^*_G(\mc N) + 1\right)^n =  (d_A  \max \operatorname{eig} \Phi^+_{\mc N})^n.
\end{equation}
In particular, when $d_A \leq 3$, it holds that
\begin{equation}
    \mc R_G(\mc N^{\otimes n}) + 1 = \left(\mc R_G(\mc N) + 1\right)^n =  (d_A  \max \operatorname{eig} \Phi^+_{\mc N})^n,
\end{equation}
or equivalently $D_{\max}(\mc N^{\otimes n}) = n D_{\max}(\mc N)$.
\end{corollary}


\section{Experiment demonstrations on the IBM cloud}

\subsection{Robustness of three example type of channels}

Here we consider the robustness of the dephasing channels $\Delta_p(\rho)=p\rho+(1-p)Z\rho Z$, stochastic damping channels $\mathcal D_p(\rho) = p\rho+(1-p)|{0}\rangle\langle{0}|$, and erasure channels $\mathcal E_p(\rho) = p\rho+(1-p)|{2}\rangle\langle{2}|$ with $\ket{2}$ orthogonal to the computational basis $\{\ket{0},\ket{1}\}$. 
We implement the dephasing channel, erasure channel, and the stochastic damping channel in Fig.~\ref{Fig:dephasing}, \ref{Fig:erasure}, and \ref{Fig:damping}, respectively. 
For these three different channels, we choose four different noise levels as $p=1,3/4,1/2,1/4$ and implement it on the IBM cloud. 
The game for the dephasing channel is the same one for the depolarising channel defined in Eq.~\eqref{Eq:innputsapp} and Eq.~\eqref{Eq:coedeplor}. The game for the erasure channel is defined in Eq.~\eqref{Eq:innputsera} and Eq.~\eqref{Eq:coeferas}. We use two qubits to encode $\ket{0}$, $\ket{1}$ and $\ket{2}$, via $\ket{0}\otimes\ket{0}\rightarrow \ket{0}$,
$\ket{0}\otimes\ket{1}\rightarrow \ket{1}$, $\ket{1}\otimes I\rightarrow \ket{2}$.
The game for the stochastic damping channel is defined in Eq.~\eqref{Eq:inputsampli} and Eq.~\eqref{Eq:coefampl}. The first qubit is used to choose the original state or $\ket{0}$.
We use the second qubit to present the input state and use the third one to replace the original state with $\ket{0}$. 
We simulate this channel via collecting the post-processing statistics. When the outcome of first qubit is $\ket{0}$, we only care about the outcomes of second qubit. Otherwise, we focus on the outcomes of third one. The experiment results of the three channels are shown in Table~\ref{Table:resultsDep}, \ref{Table:resultsEra}, and \ref{Table:resultsDam}.



\begin{figure}[h]
\centering
\begin{align*}
\Qcircuit @C=1em @R=.7em {
&\push{\ket{0}_E}&\gate{R^Y_{\theta}}&\ctrl{1}&\qw\\
\lstick{\rho_A}&\qw&\qw&\gate{Z}&\qw
\gategroup{1}{2}{2}{4}{.7em}{--}
}
\end{align*}
\caption{Implementation of the dephasing channnel. We input an ancillary state $\ket{0}_E$, rotate it with $R^Y_{\theta}=\exp(-i\theta Y/2)$, and apply a controlled-Z. Here $\theta=2\arccos(\sqrt{p})$ and $Y$ is the Pauli-Y matrix.} \label{Fig:dephasing}
\end{figure}
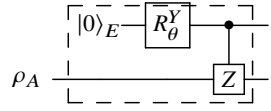


\begin{figure}[h]
\begin{align}
\Qcircuit @C=0.5em @R=0.8em {
\lstick{\ket{0}_E}& \gate{R^Y_{\theta}}  & \ctrl{1} & \qw  & \meter\\
\lstick{\ket{0}} &\qw   &    \targ  & \gate{X}  & \meter \\
\lstick{\ket{0}} &\gate{U_{in}} &\qw & \gate{M}   & \meter \\
}
\end{align}
\caption{Implementation of the erasure channel. We input an ancillary state $\ket{0}_E$, rotate it with $R^Y_{\theta}=\exp(-i\theta Y/2)$, and apply a controlled-X. Here $\theta=2\arccos(\sqrt{p})$ and $X$, $Y$ are the Pauli-X, Pauli-Y matrices, respectively. $U_{in}$ prepares the input states and $M$ changes the measurement basis. }
\label{Fig:erasure}
\end{figure}
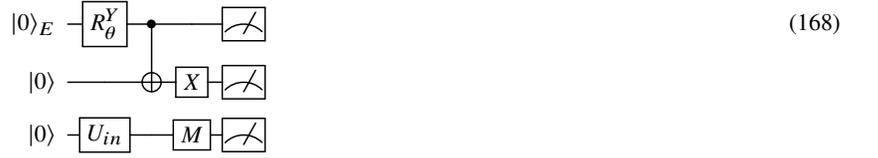

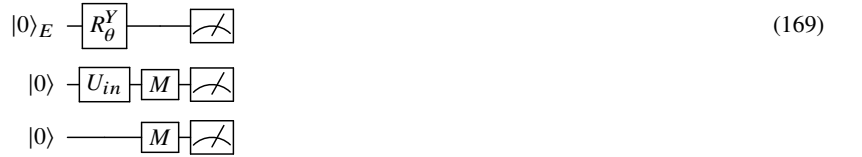
\begin{figure}[h]
\begin{align}
\Qcircuit @C=0.5em @R=0.8em {
\lstick{\ket{0}_E}& \gate{R^Y_{\theta}}   & \qw  & \meter\\
\lstick{\ket{0}} &\gate{U_{in}}   & \gate{M}  & \meter \\
\lstick{\ket{0}} &\qw & \gate{M}   & \meter \\
}
\end{align}
\caption{Implementation of the stochastic damping channel. We input an ancillary state $\ket{0}_E$, rotate it with $R^Y_{\theta}=\exp(-i\theta Y/2)$. Here $\theta=2\arccos(\sqrt{p})$. $U_{in}$ prepares the input states and $M$ changes the measurement basis. }
\label{Fig:damping}
\end{figure}

\begin{table}[h]
\centering
\centering
\begin{tabular}{c|ccc|ccc|ccc|ccc}
  \hline
 \multirow{2}*{Setting} & \multicolumn{3}{c}{p=1/4} & \multicolumn{3}{c}{p=1/2} & \multicolumn{3}{c}{p=3/4}& \multicolumn{3}{c}{p=1}\\[1mm]
 ~&IBMQ &QASM & Theory & IBMQ &QASM & Theory &IBMQ &QASM & Theory &IBMQ &QASM & Theory \\[1mm]
  \hline
  $\sigma_1$, $O_1$ & 0.2937& 0.2495 & 0.2500 & 0.5137& 0.5088& 0.5000 &0.7682&0.7471&0.7500& 0.9595&1.0000&1.0000 \\[1mm]
     $\sigma_2$, $O_2$ & 0.7261 & 0.7436 & 0.7500 ~&0.5009 &0.4993&0.5000&0.2668&0.2533&0.2500&0.0532&0.0000&0.0000\\[1mm]
   $\sigma_3$, $O_3$ & 0.9647 & 1.0000 & 1.0000 ~&0.9735 &1.0000& 1.0000&0.9720&1.0000&1.0000&0.9736&1.0000&1.0000 \\[1mm]
    $\ket{0}$, $O_4$ &0.7904& 0.7778 &0.7887 ~&0.7723 &0.7905&0.7887&0.7468&0.7834&0.7877&0.7222&0.7800&0.7887\\[1mm]
    $\ket{1}$, $O_4$ & 0.2008 & 0.2111 &0.2113 ~&0.2599 &0.2148&0.2113&0.3126&0.2139&0.2113&0.3380&0.2089&0.2113 \\[1mm]
 \hline
 Score& 0.5400& 0.5138&0.5000&0.9584&1.0049 &1.0000 &1.4219 &1.4961 &1.5000 &1.8278 &1.9924 &2.0000 \\[1mm]
  \hline
\end{tabular}
\caption{Dephasing channel}
\label{Table:resultsDep}
\end{table}

\begin{table}[h]
\centering
\centering
\begin{tabular}{c|ccc|ccc|ccc|ccc}
  \hline
 \multirow{2}*{Setting} & \multicolumn{3}{c}{p=1/4} & \multicolumn{3}{c}{p=1/2} & \multicolumn{3}{c}{p=3/4}& \multicolumn{3}{c}{p=1}\\[1mm]
 ~&IBMQ &QASM & Theory & IBMQ &QASM & Theory &IBMQ &QASM & Theory &IBMQ &QASM & Theory \\[1mm]
  \hline
  $\sigma_1$, $O_1$ & 0.2671& 0.2499 &0.2500  &0.4946&0.5046&0.5000&0.7301&0.7510&0.7500&0.9452&1.0000&1.0000 \\[1mm]
     $\sigma_2$, $O_2$ & 0.0073 & 0.0000 &0.0000 ~&0.0190 &0.0000&0.0000&0.0226&0.0000&0.0000&0.0374&0.0000&0.0000\\[1mm]
   $\sigma_3$, $O_3$ &0.2732 &0.2523&0.2500&0.4978 &0.5066&0.5000&0.7366&0.7480&0.7500&0.9403&1.0000&1.0000 \\[1mm]
    $\ket{0}$, $O_4$ & 0.2153& 0.2011 &0.1972 ~&0.3979 &0.3934&0.3943&0.5721&0.5900&0.5915&0.7394&0.7915&0.7887\\[1mm]
    $\ket{1}$, $O_4$ & 0.0641 & 0.0521 &0.0528 ~&0.1216 &0.1055&0.1057&0.1780&0.1527&0.1585&0.2205&0.20789&0.2113 \\[1mm]
 $\ket{0}$, $\ket{2}$ & 0.6655&0.7448 &0.7500 ~&0.4419&0.5043&0.5000&0.2181&0.2491&0.2500&0.0118&0.00000&0.0000\\[1mm]
    $\ket{1}$, $\ket{2}$ & 0.6681 & 0.7490 &0.7500 ~&0.4398 &0.4987&0.5000&0.2134&0.2494&0.2500&0.0116&0.0000&0.0000 \\[1mm]
 \hline
 Score& 1.2463 &1.2462 &1.2500 &1.4486&1.5136&1.5000&1.6894&1.7546 &1.7500&1.8845&2.0005&2.0000\\[1mm]
  \hline
\end{tabular}
\caption{Erasure channel}
\label{Table:resultsEra}
\end{table}

\begin{table}[h]
\centering
\centering
\begin{tabular}{c|ccc|ccc|ccc|ccc}
  \hline
 \multirow{2}*{Setting} & \multicolumn{3}{c}{p=1/4} & \multicolumn{3}{c}{p=1/2} & \multicolumn{3}{c}{p=3/4}& \multicolumn{3}{c}{p=1}\\[1mm]
 ~&IBMQ &QASM & Theory & IBMQ &QASM & Theory &IBMQ &QASM & Theory &IBMQ &QASM & Theory \\[1mm]
  \hline
  $\sigma_1$, $O_1$ & 0.6351& 0.6265 &0.6250  &0.7439&0.7491&0.7500&0.8623&0.8746&0.8750&0.9694&1.0000&1.0000 \\[1mm]
     $\sigma_2$, $O_2$ & 0.3811& 0.3857 &0.3750 ~&0.2844 &0.2471&0.2500&0.1615&0.1229&0.1250&0.0474&0.0000&0.0000\\[1mm]
   $\sigma_3$, $O_3$ &0.9866 &  1.0000 &1.0000 ~&0.9793 &1.0000&1.0000&0.9796&1.0000&1.0000&0.9742&1.0000&1.0000 \\[1mm]
    $\sigma_4$, $O_4$ &0.3135 & 0.2438 &0.2500 ~&0.5387 &0.5009&0.5000&0.7563&0.7533&0.7500&0.9668&1.0000&1.0000 \\[1mm]
    $\ket{0}$, $O_5$ & 0.4740&0.4988 &0.5000 ~&0.4639 &0.5043&0.5000&0.4375&0.4952&0.5000&0.4264& 0.4917  &0.5000\\[1mm]
    $\ket{1}$, $O_5$ & 0.5287 & 0.5026 &0.5000
    ~&0.5548 &0.5093&0.5000&0.5781&0.4963&0.5000&0.6006&  0.4998 &0.5000 \\[1mm]
 \hline
 Score&1.0699&1.0695&1.0757&1.2570&1.3025	&1.309& 1.5667&	1.6443&1.6353 &1.8734	& 2.0060 &2.0000\\[1mm]
  \hline
\end{tabular}
\caption{Stochastic damping channel}
\label{Table:resultsDam}
\end{table}

\subsection{Robustness of two qubit and three qubit gates}
We also measured the robustness of two qubit and three qubit gates on the IBM quantum cloud. Denote the two qubit controlled-X (CX) gate as $CX^0_1$ with 0 denoting the controlled qubit and 1 denoting the target qubit. We consider a sequence of CX gates with interchanged control and target qubits for two adjacent gates as
\begin{equation}
\begin{aligned}
    U_1 &= CX^0_1,\\
    U_2 &= CX^1_0CX^0_1,\\
    U_3 &= CX^0_1CX^1_0CX^0_1,\\
    U_4 &= CX^1_0CX^0_1CX^1_0CX^0_1,\\
    U_5 &= CX^0_1CX^1_0CX^0_1CX^1_0CX^0_1,\\
    U_6 &= 
    CX^1_0CX^0_1CX^1_0CX^0_1CX^1_0CX^0_1.
\end{aligned}
\end{equation}
Note that the gate $U_3$ is the swap gate and $U_6$ is the identity gate. We also considered the three qubit gate that prepares the GHZ state
\begin{equation}
    U_{GHZ}={CX}^0_2\cdot {CX}^0_1.
\end{equation}
For each gate $U$, the corresponding Choi state is $\ket{\Phi^+_{U}}=U\ket{\Phi^+}$ with $\ket{\Phi^+}=1/\sqrt{d}\sum_i \ket{ii}$ being the maximally entangled state and $d$ being the dimension. We can define a witness as $W=d{\Phi^+_{U}}$ to lower bound the robustness. Such a witness corresponds to the case where we assume  the noise is deploarising or dephasing. We can accordingly change the witness if we estimate that the gate noise is other types such as erasure or stochastic damping. According to Lemma~\ref{lemma4}, the robustness of a noisy gate $\tilde{\mc U}$ is lower bounded as
\begin{equation}
    \mc R(\tilde{\mc U}) \ge \tr[W \Phi^+_{\tilde{\mc U}}] - 1 = d\tr[\Phi^+_{ U}\Phi^+_{\tilde{\mc U}}] - 1. 
\end{equation}
Note that $\tr[\Phi^+_{ U}\Phi^+_{\tilde{\mc U}}]$ corresponds to the gate fidelity between the noisy gate $\tilde{\mc U}$ and the target gate ${\mc U}$.
To measure $\tr[\Phi^+_{ U}\Phi^+_{\tilde{\mc U}}]$, we decompose $\Phi^+_{ U}$ into a linear sum of local Pauli operators.
Since we are considering small gates, the decomposition only has a small number terms so that they are directly measured in the experiment. 

When considering large quantum gates, we can make use of the technique introduced in Ref.~\cite{PhysRevLett.106.230501} to efficiently measure the gate fidelity quantity. This requires to measure expectation values of a constant number of local Pauli operators, which are selected at random according to a weighting determined by $\Phi^+_{ U}$. 
We leave the implementation of random measurement of large quantum gates for future works. 

\end{document}